\documentclass[a4paper, 11pt]{preprint}
\pdfoutput=1 %

\usepackage{tikz}
\usetikzlibrary{calc}
\usetikzlibrary{positioning}
\usetikzlibrary{patterns, patterns.meta}

\usepackage{pgfplots}
\pgfplotsset{compat=1.11}
\usepgfplotslibrary{fillbetween}

\usepackage[cachedir=unif-minted, frozencache=true]{minted}
\colorlet{CodeBackground}{black!5!white} %
\setminted{encoding=utf8, mathescape, labelposition=topline, frame=lines,
  fontsize=\footnotesize, xleftmargin=2em, xrightmargin=2em}
\setminted[macaulay2]{label={\TT{Macaulay2}}}
\setminted[mathematica]{label={\TT{Mathematica}}}

\usepackage{multibib}
\newcites{soft}{Mathematical software}

\usepackage{caption}
\usepackage{subcaption}
\usepackage{wrapfig}

\usepackage{dsfont}
\renewcommand{\BB}[1]{\mathds{#1}}

\usepackage{sansmath}
\newcommand{\SI}[1]{\mathord{\text{\sansmath$#1$}}}

\newcommand{\NN}{\BB N}
\newcommand{\ZZ}{\BB Z}
\newcommand{\QQ}{\BB Q}
\newcommand{\RR}{\BB R}
\newcommand{\GG}{\BB G}
\NewDocumentCommand{\FF}{e{_}}{\BB F\IfValueT{#1}{_{\mkern-3mu#1}}}
\NewDocumentCommand{\Rats}{e{_}m}{#2\IfValueT{#1}{(#1)}}
\NewDocumentCommand{\Affine}{e{^}}{\BB A^{\mkern-3mu#1}}
\newcommand{\card}[1]{\lvert#1\rvert}

\newcommand{\Oh}{\CC O}

\NewDocumentCommand{\Ent}{e{_}}{\SI H^*\IfValueT{#1}{_{\mkern-3mu#1}}}
\NewDocumentCommand{\Aent}{e{_}}{\ol{\SI H^*\IfValueT{#1}{_{\mkern-3mu#1}}}}
\NewDocumentCommand{\Aalg}{e{_}}{\ol{\SI A^*\IfValueT{#1}{_{#1}}}}

\let\gcd\undefined
\DeclareMathOperator{\gcd}{gcd}

\DeclareMathOperator{\im}{im}
\DeclareMathOperator{\rank}{rk}

\DeclareMathOperator{\id}{id}

\TodoColor{Outline}{blue}

\title[The entropy profiles of a definable set over finite fields]{%
  The entropy profiles of\\
  a definable set over finite fields
}

\author{Tobias Boege}
\address{Department of Mathematics and Statistics, UiT The Arctic University of Norway, Tromsø, Norway}
\email{post@taboege.de}

\date{\today}



\begin{document}

\begin{abstract}
A definable set $X$ in the first-order language of rings defines a family of
random vectors: for each finite field $\FF_q$, let the distribution be supported
and uniform on the $\FF_q$-rational points of~$X$. We employ results from the
model theory of finite fields to show that their entropy profiles settle into
one of finitely many stable asymptotic behaviors as $q$~grows. The~attainable
asymptotic entropy profiles and their dominant terms as functions of~$q$ are
computable.
This~generalizes a construction of Matúš which gives an information-theoretic
interpretation to algebraic~matroids. %
\end{abstract}

\maketitle

\section{Introduction}
\label{sec:Intro}

Let $\xi = (\xi_i : i \in N)$ be a random vector indexed by a finite set~$N$.
All random variables in this paper are \emph{finite}, i.e., they take only
a finite number of values. %
The \emph{entropy profile} of $\xi$ is the set function $h_\xi\colon 2^N \to \RR$
associating to each $I \subseteq N$ the Shannon entropy of the marginal
distribution $\xi_I = (\xi_i : i \in I)$. %
Entropy measures the average amount of surprise upon observing the value
of a random variable. The~entropy profile of a random vector is a snapshot
of its ``information-theoretic characteristics''. Several important qualities
of how the components of~$\xi$ interact may be deduced from this vector of
$2^N$ real numbers, most notably:
\begin{itemize}[leftmargin=3em, rightmargin=0em]
\item A subvector $\xi_I$ is \emph{functionally dependent} on another $\xi_K$
  if and only if $h_\xi(I\cup K) = h_\xi(K)$.
\item Subvectors $\xi_I$ and $\xi_J$ are \emph{conditionally independent} given
  $\xi_K$ (where $I,J,K$ are mutually disjoint) if and only if $h_\xi(I\cup K) +
  h_\xi(J\cup K) = h_\xi(I\cup J\cup K) + h_\xi(K)$.
\end{itemize}
Informally, the functional dependence of $\xi_I$ on $\xi_K$ means that the
value of $\xi_I$ is almost surely determined by the value of $\xi_K$, since
the outcomes of $\xi_{I\cup K}$ are no more surprising than those of $\xi_K$~alone.
This is a strong form of dependence in~$\xi$. The~conditional independence of
$\xi_I$ and $\xi_J$ given $\xi_K$ means that whenever the outcome of $\xi_K$
is known, knowing the value of $\xi_I$ reveals no additional information about
the value of~$\xi_J$ (and vice versa).

A~variety of applications deals with probability distributions only through
their entropy profiles:
\begin{paraenum}
\item
The statistical models studied in graphical modeling~\cite{HandbookGraphical}
and causality~\cite{PearlCausality}, such as Bayesian networks, are defined
implicitly by conditional independence assumptions (which are, in turn,
derived from a~graph using certain combinatorial rules).

\item
In cryptography, it is common to model the transactions in a cryptographic
protocol using random variables and to formulate notions of (information-theoretic)
security using functional dependence and conditional independence~\cite{InfoCrypto}.
For example, the goal of secret sharing is to devise general schemes for
distributing functions $s_p$ of a secret $s$ to each participant $p \in N$
such that only preselected ``qualified'' subsets $A \in \SR Q \subseteq 2^N$
can recover the secret (in the sense that $s$ is functionally dependent on
$s_A = (s_p : p \in A)$) and all other subsets $B \notin \SR Q$ learn nothing
about the true value of $s$ (meaning $s$ is independent of~$s_B$).

\pagebreak %

\item
Several important quantities in information theory are defined in terms of
linear optimization problems over entropy profiles. This includes the classical
topic of channel capacities~\cite{CoverThomas} and various common information
measures. For instance, the Gács--Körner common information of jointly
distributed~$(\xi,\eta)$ is the maximal entropy of a random variable $\zeta$
which is simultaneously a function of~$\xi$ and a function of~$\eta$.
\end{paraenum}

What these examples hint at is an idea of ``synthetic geometry for random
variables'' in which a set of joint probability distributions is specified
by information-theoretic ``special position'' assumptions on its components.
This modeling language includes functional dependence and conditional
independence predicates which --- with enough goodwill --- resemble
parallelity and special position in geometry. It is instructive to read
the definitions of functional dependence and conditional independence above
with random variables replaced by linear subspaces, random vectors by spans
of subspaces and entropy by dimension. This is a recurring motive in the
works of the late František Matúš who discovered numerous concrete parallels
and connections between conditional independence structures and matroid
theory \cite{MatusMatroids,MatusTwocon,MatusPrinc,MatusAlgebraic}.

In applications, the number $n$ of components of $\xi$ is always fixed as
it represents the number of observables in a statistical model, participants
in secret sharing or nodes in a communication network. The natural ambient
space to study random variables in information-theoretic special position
is the set of all entropy profiles of random vectors of fixed length~$n$.
This set is known as the \emph{entropy region}~$\Ent_n$. It is naturally
embedded in $\RR^{2^n}$ by viewing each $h_\xi$ as a vector of $2^n$ real
numbers. The entropy region is quite delicate, especially on its boundary ---
and every special position assumption puts an entropy profile on the boundary.
But by a small miracle the closure $\Aent_n$ in the euclidean topology of
$\RR^{2^n}$ is a convex cone. The importance of the entropy region and its
geometric structure rests on this insight which is due to Zhang and Yeung~\cite{ZhangYeungCond}.

For practical purposes, it would be helpful to have a finite implicit
description of~$\Aent_n$ by inequalities. As a closed convex cone,
$\Aent_n$ is completely described by its dual cone, whose elements
are linear functionals $\alpha$ with $\alpha(h) \ge 0$ for every
$h \in \Aent_n$. These functionals are known as \emph{linear information
inequalities}. It is known that $\Aent_n$ is a polyhedral cone for $n \le
3$, i.e., it is described by finitely many fundamental linear~inequalities.
Matúš \cite{MatusInfinite} proved that this is no longer the case for $n \ge 4$
by constructing an infinite sequence of distributions whose entropy profiles
approach the boundary of $\Aent_4$ and whose entries decay faster than any
fixed linear function, thereby proving that the boundary is curved.
This part of the boundary was later described by a single \emph{quadratic}
information inequality by Chan and Grant~\cite{ChanGrant}.
With polyhedral descriptions ruled out, the next best result to hope for
is a semialgebraic description using finitely many polynomial inequalities.
The question whether $\Aent_4$ is semialgebraic is still open and is the
true motivation for this paper. Building on extensive computations of
Doughery, Freiling and Zeger \cite{DFZ11} which revealed an exponential
information inequality, Gómez, Mejía and Montoya \cite{GMM} devised a
strategy to disprove semialgebraicity for~$\Aent_4$ and thus for all~$n \ge 4$.
Their idea involves finding counterexamples to a parametric family of
linear inequalities.

There is still very little formalized knowledge about how to design
distributions to achieve prescribed information-theoretic effects,
particularly to make a given entropy functional negative~--- let alone
a parametric family of them.
The~present paper contains no insights into this problem, but appeals to
the principle of experimental mathematics: before we can learn to
\emph{design} distributions for a given purpose, we must first have a
pool of examples and understand how to \emph{test} them for the properties
of interest.
Almost all remarkable families of counterexamples by Kaced and Romashchenko~\cite{KR},
Studený \cite{CondIngleton} and the author \cite{Milan1} are hand-crafted
algebraic curves of \emph{binary} distributions found through experimentation
guided and verified by computer algebra systems.
However, binary random variables obey special conditional independence
laws \cite{SimecekDiss,MatusForks} and it stands to reason that there
exist invalid information inequalities which they cannot disprove.
An increase in the state spaces of the random variables increases
the number of parameters of the distribution exponentially, making
exploration beyond the binary realm very difficult (but not impossible
as \cite{NonDual4} demonstrates).

The topic of this paper is a class of probability distributions derived
from definable sets in the first-order language of rings. Each definable
set is specified by a formula $\varphi$ in which polynomial equations with
coefficients from a finite base field $\FF$ are combined using logical
connectives, and variables may be existentially or universally quantified.
This includes (affine) algebraic varieties, their differences and coordinate
projections. If $x_1, \dots, x_n$ are the free variables in $\varphi$ and
$\GG/\FF$ is a finite field extension, then $\varphi$ defines a subset~$\Rats_{\GG}{X}$
of the affine space~$\Rats_{\GG}{\Affine^n}$. This, in turn, gives rise to
a random vector~$\xi(\GG)$ on $\GG^n$ which is supported and uniformly
distributed on~$\Rats_{\FF}{X}$.
Results from model theory \cite{Definable} imply that the entropy profiles
$h_{\xi(\GG)}$ settle into one of finitely many asymptotic types as the
extension degree $e = [\GG:\FF]$ grows. This type is determined by the residue
class of $e$ modulo a period length~$m$ and denoted $h_{\xi}^{(k)}$ for
$k \in \ZZ/m$. The period length and the leading term of each component of
the asymptotic entropy profiles can be computed using a symbolic algorithm
based on Galois stratification~\cite{EffectiveCounting,FieldArithmetic4ed}.
This naturally yields distributions on arbitrarily large state spaces whose
sizes do not impact the complexity of computing the entropy profile.
When $V$ is an irreducible algebraic variety, then one of its asymptotic
entropy profiles is a refinement of its algebraic matroid~\cite{MatusAlgebraic}.
This link to synthetic geometry and computer algebra holds promise for a
deeper understanding of how to design counterexamples in information theory.

The inspiration for investigating this construction comes from a singular
example of this type due to Kaced and Romashchenko whose remarkable
information-theoretic properties derive from the arithmetic structure
of finite fields. \Cref{sec:Background} introduces the central objects
and questions related to the geometry of the entropy region as well as
required vocabulary from algebraic geometry. With these prerequisites,
we can frame the Kaced--Romashchenko example in our preferred way at
the end of the \namecref{sec:Background} and subsequently generalize it.
\Cref{sec:Unif} descends into a mix of algebraic geometry, number theory
and model theory to derive the main result on the computability of the
entropy profiles~$h_{\xi}^{(k)}$. The field-theoretic algorithms underlying
the computability result have, to our knowledge, never been implemented
in full generality. %
We return to information theory in \Cref{sec:Extension} and outline a
polyhedral geometry framework which makes use of lemmas from information
theory to degenerate a given entropy profile into one with more extreme
properties, following again the lead of Kaced and Romashchenko.

\section{Information inequalities and geometric configurations}
\label{sec:Background}

Throughout let $N$ denote a finite set of cardinality $n$ referred to as
the \emph{ground set}. It indexes a collection of objects under
consideration like random variables or coordinates. The powerset of $N$
is $2^N$. Usually $I, J, K$ denote subsets of~$N$. If $A$ is a set, then
$A^I$ is the set of all functions $I \to A$. Occasionally we will use
$A^n$ instead of $A^N$ for readability.
The letter $\xi$ is reserved for random variables, $V$ for varieties,
$X$ for (definable) sets and $x$ for variables. Sometimes we index them
with $N$ to emphasize that $x_N$ is a vector of variables $(x_i : i \in N)$
instead of a single variable~$x$.

\subsection{Information inequalities}

Let $\xi$ be a random variable taking values from a finite set~$Q$.
The map $x \mapsto x \log x$ is analytic on the real interval $(0,1]$ and
we extend it continuously (but not differentiably) to $[0,1]$ by setting
$0 \log 0 \defas 0$. The \emph{Shannon entropy} of $\xi$ is
\[
  \H{\xi} \defas - \sum_{a \in Q} \Pr[\xi = a] \log \Pr[\xi = a]
\]
and depends only on the \emph{support} of $\xi$, i.e., those values in $Q$
which have a positive probability.
If~$(\xi_i : i \in N)$ is a vector of jointly distributed random variables,
each ranging in a finite set $Q_i$, then each subvector $\xi_I = (\xi_i : i \in I)$,
for $I \subseteq N$, can be viewed as a single random variable with values
in $Q_I = \bigtimes_{i \in I} Q_i$ and as such the above definition of Shannon
entropy applies verbatim.
Recall that the \emph{entropy profile} of a random vector $\xi$ is the function
$h_\xi\colon 2^N \to \RR$ given by $h_\xi(I) = \H{\xi_I}$. Strictly speaking,
Shannon entropy and hence the entropy profile depend on the base of the logarithm.
We fix a base throughout this paper but its value is not important. Changing the
base to a fixed value $b$ amounts to division of the entropy profile~by~$\log b$
and we will do so explicitly when the need arises.

Considering the general theme of \Cref{sec:Intro}, we are led to the following
type of problem: optimize a linear functional over all entropy profiles~$h_\xi$,
where $\xi$ is indexed by a fixed set~$N$, subject to linear constraints on~$h_\xi$.
The collection of all entropy profiles, viewed as points in $\BB R^{2^N}$,
is known as the \emph{entropy region}~$\Ent_N$.
Hence, the problem can be formulated as
\begin{equation}
  \label{eq:inf} %
  \inf \Set{ \alpha(h) : h \in \Ent_N \cap \CC L },
\end{equation}
where $\alpha$ is a linear functional (usually non-negative on $\Ent_N$)
and $\CC L$ is a linear space in $\RR^{2^N}$ (usually cut out by the linear
equations corresponding to functional dependence and conditional independence
assumptions).
Such a problem is solved in two steps. First, in the ``achievability'' step,
a lower bound is established by explicit construction of a sequence of random
vectors~$\xi^{(k)}$. Second, in the ``converse'' step, linear inequalities
on~$\Ent_N$ (so-called \emph{information inequalities}) are combined with
the linear equations of the problem to show that the lower bound is, in fact,
also an upper bound.
This approach is so successful because of a small miracle: the closure~$\Aent_N$
of the entropy region in the euclidean topology of $\RR^{2^N}$ is a convex~cone~\cite{ZhangYeungCond}.
This cone is known as the \emph{almost-entropic region}. Its convexity ensures
that local optima are global, i.e., if the achievability step produces a sequence
of entropy profiles $h_k = h_{\xi^{(k)}} \in \Ent_N$ converging to a local
optimum with value $\alpha^* = \lim_{k \to \infty} \alpha(h_k)$,~then there
must exist a valid inequality of the form
\begin{equation}
  \label{eq:CondIneq} %
  \text{$\alpha(h) \le \alpha^*$ holds for all $h \in \Ent_N \cap \CC L$},
\end{equation}
thereby certifying the optimality of~$h_k$.

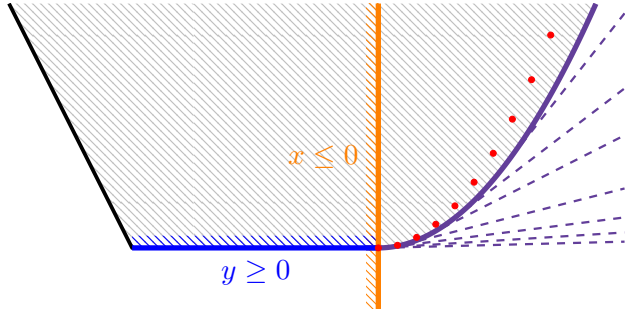
\begin{figure}
\begin{center}
\begin{tikzpicture}[cap=rounded, scale=1.3]
\begin{axis}[xmin=-0.5, xmax=5, ymin=-0.5, ymax=2, hide axis, axis equal]
  \clip (-0.5,-0.5) rectangle (5,2);
  \addplot[name path=lower, draw=none, line width=.5mm, smooth, domain=0:4.767, samples=200] {max(max(0, 2-2*x), 0.64*max(0, (x-3))^2)};
  \addplot[black, line width=.5mm, smooth, domain=0:1.1, samples=200] {max(0, 2-2*x)};
  \addplot[name path=upper, draw=none, domain=0:4.1, samples=100] {2};

  \addplot[pattern=north west lines, pattern color=gray!50]
    fill between [of=lower and upper, split, every segment no 1/.style={fill=none}];

  \foreach \i in {4, 3.6, 3.4, 3.2, 3.1, 3.05, 3.02} {
    \addplot[RoyalPurple, line width=.3mm, smooth, domain=\i:5, samples=100, dashed]
      {0.64*((\i-3)^2 + (x-\i)*2*(\i-3))};
  }
  \addplot[RoyalPurple, line width=.7mm, smooth, domain=2.9:4.767, samples=100] {0.64*max(0, (x-3))^2};

  \addplot[red, only marks, mark size=.3mm, domain=3:4.4, samples=10] {8/9*(x-3)^2};

  \draw[name path=ineqa1, draw=none, line width=.7mm] (2.9,2) -- (2.9,-0.5);
  \draw[name path=ineqa2, orange, line width=.7mm] (3,2) -- (3,-0.5)
    node [pos=0.5, left, xshift=-1mm] {$x \le 0$};
  \addplot[pattern=north west lines, pattern color=orange, line width=.5mm]
    fill between [of=ineqa1 and ineqa2];

  \addplot[name path=ineqb1, draw=none, line width=.7mm, smooth, domain=1:3.01, samples=100] {0.1};
  \addplot[name path=ineqb2, blue, line width=.7mm, smooth, domain=1:3.01, samples=100] {0}
    node [pos=0.5, below] {$y \ge 0$};
  \addplot[pattern=north west lines, pattern color=blue, line width=.5mm]
    fill between [of=ineqb1 and ineqb2];
\end{axis}
\end{tikzpicture}
\end{center}
\caption{The inequality $y = 0 \implies x \le 0$ is essentially conditional
because the slopes of the purple boundary tangent lines approach zero as
$x \to 0$. This~can be seen without knowing the exact equation of the boundary:
the red dots at $(\sfrac1{\lambda}, \sfrac8{9\lambda^2})$ lie inside the region
and violate the unconditional version $x \le \lambda y$ for any $\lambda > 0$.
They approach the boundary arbitrarily well near the extreme~point~$(0,0)$.}
\label{fig:Cond}
\end{figure}

The inequality \eqref{eq:CondIneq} is a \emph{conditional linear information
inequality} since it only holds for entropy profiles restricted to the linear
space~$\CC L$. Let $\CC L$ be the solution set to the system of linear equations
$\beta_1(h) = \dots = \beta_m(h) = 0$. A more compact and schematic way to
write~\eqref{eq:CondIneq} is
\[
  \beta_1 = \dots = \beta_m = 0 \implies \alpha \le \alpha^*,
\]
keeping in mind that this implication is only valid when evaluated at
entropy profiles $h \in \Ent_N$. One~way to prove such a conditional
inequality is by proving a stronger \emph{unconditional inequality}
\begin{equation}
  \label{eq:Uncond} %
  \alpha \le \alpha^* + \lambda_1 \beta_1 + \dots + \lambda_m \beta_m,
\end{equation}
for some Lagrange multipliers $\lambda_1, \dots, \lambda_m \in \BB R$.
If \eqref{eq:CondIneq} arises in this way from a valid unconditional
inequality~\eqref{eq:Uncond}, then it is called \emph{unconditional}
or \emph{inessential}, otherwise it is \emph{essentially conditional}.
\Cref{fig:Cond} illustrates this concept. We also refer to
\cite[Section~II.C]{KR} for more~explanations.

The most well-known class of information inequalities is named after
Claude Shannon who demonstrated their usefulness in \cite{Shannon}.
Besides the conventional normalization $h(\emptyset) = 0$ which assigns
a zero entropy to an empty random variable, the \emph{Shannon inequalities}
specify
\begin{description}[noitemsep, parsep=0.4em]
\item[Monotonicity] $h(I \cup K) \ge h(K)$, and
\item[Submodularity] $h(I \cup K) + h(J \cup K) \ge h(I \cup J \cup K) + h(K)$.
\end{description}
In other words, the entropy region is contained in the polyhedral cone
of \emph{polymatroids}. This usage of the word ``polymatroid'' for a
monotone, submodular set function is standard in information theory, but
it conflicts with combinatorial optimization where a polymatroid is a
polytope associated to such a function; cf.~\cite[Chapter~18]{Welsh}.
Note that the extreme cases of these inequalities correspond precisely
to functional dependence and conditional independence. Hence every
entropy profile can be described in an alternative coordinate system
which specifies how far away it is from satisfying certain functional
dependence or conditional independence constraints.
Because of their importance, we introduce special notation for them:
\begin{align*}
  \CId{I|K}{h} &\defas h(I\cup K) - h(K) \;\; \text{and} \\
  \CId{I:J|K}{h} &\defas h(I\cup K) + h(J\cup K) - h(I\cup J\cup K) - h(K).
\end{align*}
The Shannon inequalities can then be restated as the non-negativity of
these functionals. In the context of random variables $\CId{\xi_I|\xi_K}$
is the \emph{conditional entropy} and $\CId{\xi_I:\xi_J|\xi_K}$ the
\emph{conditional mutual information}.

The~theory of information inequalities is a fascinating and challenging
subject. As tools for solving information-theoretic problems, they have
also found use in combinatorics, graph theory and Kolmogorov complexity;
see, e.g., \cite{Zhao,KolmogorovInfo}. The large corpus of known
information inequalities serves as a testimony to the interest in
these tools. At the beginning of this enterprise is a result by Zhang
and Yeung \cite{ZhangYeungCond} who found a conditional information
inequality which they lifted in \cite{ZhangYeung} to an unconditional one.
This became the first known \emph{non-Shannon} information inequality:
a valid inequality which separates the entropy region from the enclosing
polymatroid cone.
In 2007, Matúš \cite{MatusInfinite} showed that $\Aent_n$ is not polyhedral
for $n \ge 4$ by exhibiting an infinite list of independent information
inequalities.
A major milestone is the list of Dougherty, Freiling and Zeger \cite{DFZ11}
of over 200 unconditional inequalities, a few dozen conditional ones and
several infinite families. They were derived from Shannon inequalities
using a systematic lift-and-project technique.
Kaced and Romashchenko \cite{KR} developed the theory further and gave
the first example of an essentially conditional information inequality.
Today,~Matúš's infinite list is well-understood in this framework.
Studený \cite{CondIngleton} observed that earlier work \cite{MatusI,MatusII,MatusIII}
on the characterization of representable conditional independence
structures on four random variables --- which was achieved using
considerable effort at the time --- follows easily from a class of
conditional information inequalities, the so-called \emph{CI-type
conditional Ingleton inequalities}. Those impose conditional
independence assumptions and guarantee the non-negativity of the
Ingleton functional
\[
  \Ing{A:B|C:D}{h} \defas \CId{C:D|A}{h} + \CId{C:D|B}{h} + \CId{A:B}{h} - \CId{C:D}{h}.
\]
It is well-known in matroid theory that the condition $\Ing{A:B|C:D}{h} \ge 0$
is necessary for a polymatroid $h$ to be linearly representable over a division
ring (in particular a field) \cite{Ingleton}.
The~classification of these inequalities on four random variables was
finished in~\cite{Milan1}.

An essentially conditional information inequality is a compactly encoded
stronger version of an infinite family of unconditional inequalities.
Despite the value they provide, only very few conditional information
inequalities are known to be essentially conditional. It would be
interesting to find the maximal CI-type conditions which give essentially
conditional Ingleton inequalities. Partial results towards this goal are
summarized in \cite{Milan1}. Its resolution would solve one of the most
alluring open question in this area:

\begin{named*}{{The Gómez--Mejía--Montoya problem \cite{GMM}}}
Is the information inequality
\begin{equation*}
  \CId{A:C|D} = \CId{A:D|C} = \CId{B:C|D} = \CId{B:D|C} = 0 \implies \Ing{A:B|C:D} \ge 0
\end{equation*}
essentially conditional?
\end{named*}

\begin{remark}
This question matches Question~1 in \cite{GMM}. We note that, possibly due to
a typo, the development in \cite[Section~4]{GMM} does not actually lead to this
question but a variant in which the assumption $\CId{A:D|C} = 0$ is replaced
by $\CId{C:D|A} = 0$. This discrepancy seems to be the result of an incorrect
reading of the inequality from \cite[Theorem~10]{DFZ11}.

Reading~that~inequality~as
\begin{gather*}
  \left(2^{s-1} - 1\right) \Bigg(
    \Ing{A:B|C:D} - \Big[\CId{B:C|D} + \CId{B:D|C}\Big] +
    \frac{1}{2^{s-1} - 1} \CId{C:D|A} + {} \\
    \frac{2^{s-1} (s - 1)}{2^s - 2}\Big[\CId{A:C|D} + \CId{A:D|C} + \CId{B:C|C} + \CId{B:D|C}\Big]
  \Bigg) \ge 0
\end{gather*}
and following the reasoning of \cite{GMM} leads to our formulation of the
GMM~problem above. As~Gómez, Mejía and Montoya show in \cite[Theorem~22]{GMM},
an affirmative answer to it would imply that $\Aent_n$ is not semialgebraic
for all~$n \ge 4$.
\end{remark}

A~proof of essential conditionality requires a family of entropy vectors
which violates every possible unconditional version~\eqref{eq:Uncond} of
the inequality.
If the linear functionals $\beta_i$ are non-negative on $\Ent_N$ (for
example, if they are functional dependence or conditional independence
functionals), then the problem of proving essential conditionality
simplifies slightly. Note that if there exist Lagrange multipliers
$\lambda = (\lambda_1, \dots, \lambda_m)$ which furnish a valid
unconditional version of the inequality and if $\lambda' \ge \lambda$
component-wise, then $\lambda'$ is also a valid set of multipliers.
Hence, we may assume that all multipliers $\lambda_i$ are equal
(to their maximum) and arbitrarily~large. Equivalently, \eqref{eq:CondIneq}
is essentially conditional if and only if for all $\eps > 0$ arbitrarily
small there is a distribution which violates
\begin{equation}
  \tag{$\text{\ref{eq:Uncond}}'$} \label{eq:Uncond'}
  \eps (\alpha - \alpha^*) \le \beta_1 + \dots + \beta_m.
\end{equation}
Essential conditionality proofs usually consist of a curve of random
vectors parametrized by $\delta > 0$ such that the right-hand side
of~\eqref{eq:Uncond'} tends to zero with $\delta$ faster than the
left-hand side, for any fixed $\eps > 0$. We will study an example
of this in \Cref{sec:KR}.

\subsection{Affine algebraic varieties}

To define the class of probability distributions considered in this paper,
we work with affine varieties defined over a finite base field~$\FF$ and
their rational points over large finite extensions~$\GG/\FF$. To benefit
from the apparatus of algebraic geometry, we shall pass to the direct limit
and view varieties as subsets of affine space over the algebraic closure~$\ol{\FF}$.
On the other hand, the base field makes it easy to carry out symbolic
computations and it is the arithmetic properties of \emph{finite} fields
which are paramount to our goals.
The purpose of this section is to recall some practical facts which help
bridge these two points of view. We~also fix terminology and conventions.
The referenced results can be found in \cite{LangGeometry,LangAlgebra},
\cite{Morandi} and~\cite{CLO}.

Let $\FF$ be a field with algebraic closure $\ol{\FF}$ and $x_N =
(x_i : i \in N)$ a tuple of variables indexed by a finite set~$N$ of size~$n$.
An \emph{(affine algebraic) variety} defined over~$\FF$ is the zero locus
in $\ol{\FF}^n$ of a finite number of polynomial equations $f_1, \dots, f_k
\in \FF[x_N]$, i.e.,
\[
  V = \Set{ a \in \ol{\FF}^n : f_1(a) = \dots = f_k(a) = 0 }.
\]
Contrary to the usage of the word ``variety'' in \cite{LangGeometry},
we do not require $V$ to be irreducible in the Zariski topology; instead,
we use the adjective(s) \emph{(absolutely) irreducible} explicitly.
For any intermediate field $\ol{\FF}/\GG/\FF$ the set of
\emph{$\GG$-rational points} is $\Rats_{\GG}{V} = V \cap \GG^n$.
We shall also refer to these sets as varieties.
The ideal in $\SR I$ in $\FF[x_N]$ generated by the $f_1, \dots, f_k$
carries more information about the variety than the sets of rational
points. %

The polynomial functions $\Rats_{\FF}{V} \to \FF$ form an affine
$\FF$-algebra~$\FF[V]$ called the \emph{coordinate~algebra}. It is
isomorphic to $\FF[x_N]/\SR I$ and hence generated (as a ring) by
the \emph{coordinate functions} $\ol{x_i} = x_i + \SR I \in \FF[x_N]/\SR I$.
The~\emph{dimension} $\dim(V)$ is the Krull dimension of~$\FF[V]$ and
can be computed from its Hilbert series via a Gröbner basis of~$\SR I$
\cite[Chapter~9]{CLO}. In particular it is independent of
extensions of the base field over which $V$ is defined.

As any ideal $\SR I$ has a unique decomposition into finitely many primary
ideals $\SR P_j$ in $\FF[x_N]$ by the Lasker--Noether theorem \cite[Section~4.8]{CLO},
the variety $\Rats_{\FF}{V}$ can be written uniquely as a union of distinct
$\FF$-irreducible varieties $\Rats_{\FF}{V_j}$ which are its
\emph{$\FF$-irreducible components}. The~dimension of~$V$ is the
maximum over the dimensions of the~$V_j$. If $V$ is $\FF$-irreducible
(i.e., its vanishing ideal in $\FF[x_N]$ is prime), then $\FF[V]$ is
an integral domain its field of fractions~$\FF(V)$ has transcendence
degree~$\dim(V)$ over~$\FF$.
A~variety irreducible over $\FF$ may become reducible over an extension~$\GG/\FF$.
Thus, field extensions may change the number of irreducible components but
they never change the dimension.
The irreducible components of $V$ over $\ol{\FF}$ are also called the
\emph{absolutely} irreducible components. Each of them is cut out by
a finite set of polynomials by Hilbert's Basis theorem \cite[Section~2.5]{CLO},
each with finitely many coefficients from $\ol{\FF}$. Hence, there is
already a finite extension of $\FF$ over which $V$ splits into its
absolutely irreducible components.

For $I \subseteq N$ the image of $V$ under the coordinate projection
$\pi_I\colon \ol{\FF}^N \to \ol{\FF}^I$ is generally not a variety.
If $V$ is cut out by an ideal $\SR I \subseteq \FF[x_N]$, then the
Zariski closure of $\pi_I(V)$ is cut out by the \emph{elimination ideal}
$\SR I \cap \FF[x_I]$; cf.~\cite[Section~4.4]{CLO}. A Gröbner basis of
the elimination ideal can be computed from one (in a suitable monomial
ordering) of the original ideal and hence it is possible to compute
$\dim \pi_I(V)$ for each coordinate projection~$\pi_I$.
If $\SR I$ is prime, then so is $\SR I \cap \FF[x_I]$ and the function
$I \mapsto \dim \pi_I(V)$ is (the rank function of) the \emph{algebraic matroid}
associated with~$\SR I$~\cite{AlgebraicAction}. The dimensions are then also
given by the transcendence degrees of the field extensions $\FF(\ol{x_I})/\FF$,
where $\ol{x_i} \in \FF(V)$.
A~\emph{circuit} in this matroid is an inclusion-minimal subset $C \subseteq N$
for which the coordinate functions $\ol{x_C}$ are algebraically dependent
over~$\FF$. The~elimination ideal for a circuit is not only prime but also
principal \cite[Theorem~11]{AlgebraicAction}. Up to units, the irreducible
generator of this ideal is unique and called the \emph{circuit polynomial}.

\subsection{The Kaced--Romashchenko configuration}
\label{sec:KR}

The following conditional information inequality is essentially conditional
\cite[Theorem~3]{KR}:
\begin{equation}
  \label{eq:ABC}
  \CId{A:B} = \CId{A:B|C} = 0 \implies \Ing{A:B|C:D} \ge 0.
\end{equation}
To prove this, one has to exhibit a sequence of probability distributions
$(A,B,C,D)$ depending on a parameter $\eps \to 0$ such that $\CId{A:B} +
\CId{A:B|C} + \eps \, \Ing{A:B|C:D} < 0$. If such a sequence exists, there
will even be one whose entropies are bounded by~$1$ (divide the entropy
profile by $\H{ABCD}$ and use that $\Aent_4$ is a cone to ensure the
existence of a sequence of probability distributions). Thus we may assume
that all entropy functionals are bounded. This implies that such a sequence
has to approach the conditional independence model $\CId{A:B} = \CId{A:B|C} = 0$
but since \eqref{eq:ABC} is valid, the sequence cannot be exactly on it.
The example given in \cite[Section~IV.B]{KR} stands out from all other known
essential conditionality proofs in that it is profoundly geometric. We~give
a slight variation here which better fits the narrative. Fix a finite field~$\FF_q$
(of $\text{characteristic} \neq 2$) and consider the following objects in
the affine plane~$\Rats_{\FF_q}{\Affine^2}$:
\begin{itemize}[noitemsep, parsep=0.5em]
\item Two points $A = (a_1, a_2)$ and $B = (b_1, b_2)$ with $a_1 \neq b_1$.
\item The line $C = \Set{ (x,y) \in \Rats_{\FF_q}{\Affine^2} : y = c_1 x + c_0 }$
  through $A$~and~$B$.
\item A parabola $D = \Set{ (x,y) \in \Rats_{\FF_q}{\Affine^2} : y = d_2 x^2 + d_1 x + d_0 }$
  through $A$~and~$B$ which is non-degenerate, i.e., $d_2 \neq 0$.
\end{itemize}
Choosing the $2+2+2+3 = 9$ parameters of these objects independently and
uniformly at random from the reservoir of values they are allowed to take,
this procedure defines a probability distribution in~$\FF_q^9$. This
distribution is uniform on its support, i.e., any two possible events are
equally likely. Kaced and Romashchenko's original example differs only
in that they choose the line $C$ first from among all non-vertical lines,
then choose $A$ and $B$ on it, allowing them to coincide but stipulating
that in this case $C$ must be tangent to~$D$.

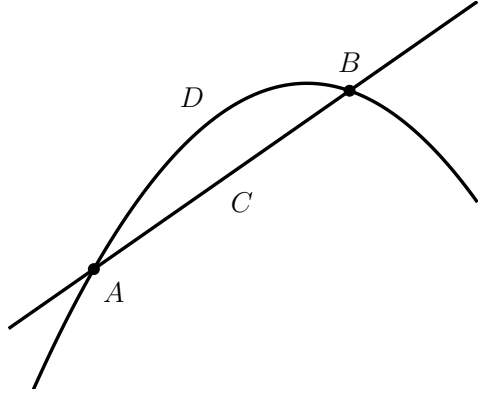
\begin{figure}
\begin{center}
\begin{tikzpicture}[cap=rounded, scale=0.9]
\begin{axis}[xmin=0.0, xmax=5.5, ymin=0.5, ymax=3.75, hide axis]
  \addplot[black, line width=.5mm, smooth, domain=0.0:5.5, samples=100] {1.75*x - 0.25*x^2};
  \addplot[black, line width=.5mm, smooth, domain=0.0:5.5, samples=100] {1.0 + 0.5*x};
  \draw[fill, draw=black] (1,1.5) circle (.8mm) node [below right] {\large\strut$A$};
  \draw[fill, draw=black] (4,3.0) circle (.8mm) node [above] {\large\strut$B$};
  \node[below right] at (2.5,2.25) {\large\strut$C$};
  \node[above left] at (2.4,2.7) {\large\strut$D$};
\end{axis}
\end{tikzpicture}
\end{center}
\caption{A generic point of the Kaced--Romashchenko variety over~$\RR$.}
\label{fig:KR}
\end{figure}

This variety is 5-dimensional but a generic point on it can be visualized
as the configuration of a parabola, a line and two intersection points in
the plane, as in \Cref{fig:KR}.
The entropy functionals appearing in \eqref{eq:ABC} are computed as follows:
\begin{paradesc}
\item[$\CId{A:B}$]
The possible values for $A$ are uniformly distributed in $\FF_q^2$, thus
$\H{A} = 2 \log q$. Given $B$, the choices for $A$ are reduced to $q(q-1)$
but $A$ still ranges uniformly in them, so $\H{A|B} = \log q + \log(q-1)$.
This computes $\CId{A:B} = \H{A} - \H{A|B} = \log(q) - \log(q-1)$.
As $q$ grows large, the dependence between $A$ and $B$ will become smaller
and smaller but it never vanishes exactly.

\item[$\CId{A:B|C}$]
Given $C$, the reservoir of possible values for $A$ and $B$ are reduced by
a factor of $q$, but the analysis is the same: $\CId{A:B|C} = \H{A|C} - \H{A|BC}
= \log(q) - \log(q-1)$.

\item[$\CId{C:D|A}$ and $\CId{C:D|B}$]
Now consider $A$ known. There are $q(q-1)$ equally likely non-degenerate
parabolas through $A$, so $\H{D|A} = \log(q) + \log(q-1)$. If the line $C$
is also given, $D$ ranges in the non-degenerate parabolas through $A$ which
do not have tangent line $C$ at~$A$. This subtracts $q$ parabolas leading
$\H{D|AC} = \log(q) + \log(q-2)$. Thus $\CId{C:D|A} = \log(q-1) - \log(q-2)$
and the value of $\CId{C:D|B}$ is the same.

\item[$\CId{C:D}$]
As $C$ ranges uniformly in all non-vertical lines, $\H{C} = 2\log(q)$.
Now let $D$ be fixed and consider all possible lines~$C$ which intersect
$D$ in two distinct $\FF_q$-rational points. Equivalently, the discriminant
$\Delta = (d_1 - c_1)^2 - 4\,d_2(d_0 - c_0)$ is a non-zero square in~$\FF_q$.
Since the parameters of $D$ are fixed, we view $\Delta$ as a function of
$(c_0, c_1) \in \FF_q^2$. By fixing $c_1$, this becomes an affine-linear
function, and from there it is easy to see that $\Delta(c_0, c_1)$
attains every value in $\FF_q$ exactly $q$~times. Since the set of non-zero
squares in $\FF_q$ has cardinality exactly $\frac{q-1}{2}$, we finally
deduce $\H{C|D} = \log\left(\frac{q(q-1)}{2}\right)$ and hence
$\CId{C:D} = \log(q) - \log(q-1) + \log(2)$.
\end{paradesc}

In total, this shows
\begin{equation}
  \CId{A:B} + \CId{A:B|C} + \eps \, \Ing{A:B|C:D} =
  2\log\left(\frac{q}{q-1}\right) + \eps\,\left(
  2\log\left(\frac{q-1}{q-2}\right) - \log2 \right).
\end{equation}
For any $\eps > 0$, this quantity becomes negative as $q \to \infty$.
This provides the desired sequence of distributions and proves that
\eqref{eq:ABC} is essentially conditional. On this sequence, the
conditional mutual informations appearing in \eqref{eq:ABC} vanish
asymptotically, except for $\CId{C:D}$ which has an absolute $\log 2$
summand. This dependence is explained as follows: the fact that $C$
and $D$ intersect means that $D$ has rational points and is thus a
special kind of parabola. This ``one bit'' of mutual information is
still vanishingly small compared to the marginal entropies $\H{A},
\H{B}, \H{C}, \H{D}$ which tend to infinity proportionally to~$\log q$.
If the entropy is normalized by choosing $q$ as the basis of the
logarithm, then $\CId{C:D}$ will also vanish --- just at a slower rate
than the other mutual informations, making it possible to violate
unconditional versions of \eqref{eq:ABC}.

This analysis is relatively simple, for two reasons. First, the
joint distribution and the examined marginals and conditionals above
are uniform on their respective supports, so entropy merely represents
the size of their supports. Second, the support of the joint distribution
is a (quasi-affine) variety in $\FF_q^9$ defined by the following conditions:
\begin{gather}
  \label{eq:KReqs}
  \begin{gathered}
  a_2 = c_1 a_1 + c_0 = d_2 a_1^2 + d_1 a_1 + d_0, \\
  b_2 = c_1 b_1 + c_0 = d_2 b_1^2 + d_1 b_1 + d_0, \\
  a_1 \neq b_1, \; d_2 \neq 0.
  \end{gathered}
\end{gather}
This construction of equipping a variety over a finite field with the
uniform distribution was considered by František Matúš in one of his
last papers~\cite{MatusAlgebraic}. Let $V$ be an irreducible affine
variety defined over a finite field~$\FF$ and for any field extension
$\GG/\FF$ consider the distribution on~$\GG^n$ uniform and supported
on~$V(\GG)$. The coordinate functions are the components of a random
vector $(\xi_1, \dots, \xi_n)$. Using the Lang--Weil bound, Matúš proved
that for every subset $I$ of these components, the marginal entropy satisfies
\[
  \frac1{\log \card{\GG_n}} \H{\xi_i : i \in I} \to \dim \pi_I(V), \;
  \text{for some tower of fields $\FF = \GG_0 \subseteq \GG_1 \subseteq
  \dots$},
\]
where $\pi_I$ is the projection onto the coordinates~$I$. The limit point
$(\dim \pi_I(V) : I \subseteq N)$ is the \emph{algebraic matroid} of~$V$
and this result shows that algebraic matroids are almost-entropic.

The variety $V$ defined by the equations in \eqref{eq:KReqs} is reducible.
The non-degeneracy conditions $a_1 \neq b_1$ and $d_2 \neq 0$ remove all
but one of the irreducible components and leave behind a Zariski-open set
whose closure is irreducible, smooth and of dimension~$5$.
These properties are true over all fields of large enough characteristic.
The points gained in the closure are contained in a lower-dimensional
variety and thus do not contribute to the entropy profile in a significant~way,
as~we shall see in \Cref{sec:Unif}.

The~Kaced--Romashchenko distributions are thus essentially derived from
the uniform distributions supported on the rational points~$V(\FF_q)$.
The last step in making this connection consists of a simple combinatorial
operation on entropy profiles known as \emph{factoring}. Let $h\colon 2^M \to \RR$
be any set function. Any partition $M = \bigsqcup_{i \in N} M_i$ induces
a map $\varrho\colon 2^N \to 2^M$ via $\varrho(I) = \bigcup_{i \in I} M_i$.
The \emph{factor} of $h$ by~$\varrho$ is then simply the pullback
$\varrho^* h\colon 2^N \to \RR$ defined by $\varrho^* h(I) = h(\varrho(I))$.
If $h$ is the entropy profile of random variables $(X_i : i \in M)$, then
$\varrho^* h$ is the entropy profile of the random vector whose entries are
vector-valued random variables $(X_{\varrho(i)} : i \in N)$.
Thus, factoring preserves entropicness.
The right grouping turns the $9$ coordinate functions $(a_1, a_2, b_1, b_2,
c_0, c_1, d_0, d_1, d_2)$ of~$V(\FF_q)$ into the four geometric objects~$(A, B, C, D)$
described by Kaced and Romashchenko.

The algebraic matroid of the Kaced--Romashchenko configuration is easy to
compute using elimination theory in polynomial rings and Gröbner bases as
shown below using \Macaulay2~\citesoft{M2}:
\begin{minted}{macaulay2}
R = QQ[a0,a1, b0,b1, c0,c1, d0,d1,d2];
I = ideal(
  c0 + c1*a0 - a1,           -- A on C
  c0 + c1*b0 - b1,           -- B on C
  d0 + d1*a0 + d2*a0^2 - a1, -- A on D
  d0 + d1*b0 + d2*b0^2 - b1  -- B on D
);
I = saturate(I, (a0-b0)*d2);
-- This ideal is prime, smooth and of dimension 5.
print { isPrime I, mingens ideal singularLocus I, dim I };

-- Compute the factor of the algebraic matroid.
rho = new HashTable from {
  "A" => {a0,a1}, "B" => {b0,b1}, "C" => {c0,c1}, "D" => {d0,d1,d2}
};
h = new HashTable from apply(subsets(keys rho), L -> (
  v := flatten(apply(L, l -> rho#l));
  sort(toList(set(keys rho) - set(L))) =>
    dim R - #v - codim eliminate(v, I)
));
print toString h;
-- new HashTable from { {} => 0,
--   {A} => 2, {B} => 2, {C} => 2, {D} => 3,
--   {A, B} => 4, {A, C} => 3, {A, D} => 4, {B, C} => 3, {B, D} => 4, {C, D} => 5,
--   {A, B, C} => 4, {A, B, D} => 5, {A, C, D} => 5, {B, C, D} => 5,
--   {A, B, C, D} => 5 }
\end{minted}
\pagebreak %
\noindent
However, it must be emphasized that the algebraic matroid does not
provide a proof that \eqref{eq:ABC} is essentially conditional because
this limit point of the entropy profiles satisfies $\CId{A:B} = \CId{A:B|C}
= \Ing{A:B|C:D} = 0$. The~information about how quickly each of these
quantities vanishes (and that $\Ing{A:B|C:D}$ indeed approaches zero
from \emph{below}) is lost.

\section{Computing entropy profiles of a definable set}
\label{sec:Unif}

The entropy of the random vector $\xi$ supported uniformly on the
$\FF_q$-rational points of a variety~$V$ is precisely $\log \card{V(\FF_q)}$.
Hence, the computation of this entropy is equivalent to counting rational
points on~$V$, which is an important task in number theory. The argument
in \Cref{sec:KR} also requires $q$ to be sufficiently large and therefore
the field cannot be fixed a~priori. The solution to this problem is the
(local) Hasse--Weil zeta function of~$V$ which encodes the exact point
counts of $V$ over the finite extensions of~$\FF_q$. Since it is a rational
function, this information can be finitely represented. However, the zeta
function is very hard to compute.
Fortunately, the emphasis of the computations in \Cref{sec:KR} is not
on the exact point count over $\FF_q$ as much as its asymptotics over
larger and larger field extensions. For information-theoretic purposes,
the Lang--Weil bound~\cite{LangWeil} $\card{V(\FF_q)} = q^{\dim V} +
\Oh(q^{\dim V - \sfrac12})$ is sufficient as it provides the dominant
term as a function of~$q$ (provided that $V$ is absolutely irreducible
over~$\FF_q$). Under a logarithm with base~$q$, the terms of lower order
are asymptotically absorbed and $\H{\xi} \approx \dim V$, an estimate
which is easy~to~compute.

Computing the entropy for proper subvectors $\xi_I$, $I \subsetneq N$,
is more difficult. The marginal distribution is supported on the image of
the coordinate projection $\pi_I(V(\FF_q))$ but this need no longer be a
variety and the distribution need not be uniform on it, as the following
example~shows.

\begin{example}[Roots of a cubic] \label{ex:Cubic}
Consider the irreducible hypersurface $V$ defined by the polynomial
$x^3 + ax^2 + bx + c \in \ZZ[a,b,c,x]$ and its projection onto~$(a,b,c)$.
The projection consists of all monic cubics which have a rational root.
The size of a fiber of this projection is the number of distinct roots
of a given cubic. Using \Macaulay2 we see empirically that this number
is not uniformly distributed:
\begin{minted}{macaulay2}
K = GF(7,3); R = K[a,b,c,x];
tally apply(1 .. 1000000, i -> (
  g := sub(x^3 + a*x^2 + b*x + c, {a=>random(K), b=>random(K), c=>random(K)});
  degree gcd(g, x^(K.order) - x) -- counts distinct roots
))
-- Tally{{0} => 333630}
--       {1} => 498671
--       {2} => 2905
--       {3} => 164794
\end{minted}
To prove it, we prefer to parametrize the cubics in the following way:
\[
  (x-\alpha) \, \left( x^2 + ax + \frac{a^2-\Delta}{4} \right),
\]
so that $\alpha$ witnesses the fact that there is a rational root and
$\Delta$ is the discriminant of the remaining quadratic polynomial.
For the cubic to have two distinct roots, the discriminant parameter~$\Delta$
must be zero which happens with probability approaching zero as $q\to \infty$.
If $\Delta$ is a non-square in $\FF_q$, then the cubic will have only one
solution. This happens in asymptotically $\sfrac12$ of the~cases.
Three solutions appear when the cubic factors completely. These cubics are
alternatively parametrized by their roots. However, this map is generally
six to one since $S_3$ acts on the roots. Hence, a randomly chosen monic
cubic over $\FF_q$ has about $\sfrac16$ chance to have three distinct~roots,
$\sfrac12$ to have one root and $\sfrac13$ to have no root.
The projection of $V$ is thus not a variety: it is full-dimensional in the
sense that its Zariski closure is the entire space, but it has a density
of~$\sfrac23$. %
\end{example}

\pagebreak %

\subsection{The uniform distribution on a definable set}

Varieties and their coordinate projections fall into the broader class of
\emph{definable sets}. We briefly explain this concept from model theory
in our setting. For an introduction to model theory and its applications
to algebra, Marker's book \cite{Marker} is highly recommended. We work in
the first-order language of rings which means that a \emph{formula} is
built up in the following way:
\begin{itemize}[leftmargin=3em, rightmargin=1em, noitemsep, parsep=0.5em]
\item The \emph{atomic} formulas are $f = 0$ for a multivariate polynomial~$f$
  with $\ZZ$-coefficients;
\item if $\varphi$ and $\psi$ are formulas, then so are $\varphi \land \psi$,
  $\varphi \lor \psi$ and $\neg\varphi$;
\item if $\varphi$ is a formula containing a free variable $x$, then
  $\exists x\colon \varphi$ and $\forall x\colon \varphi$ are formulas.
\end{itemize}
Relative to a field $\FF$, we can provide such a formula $\varphi$ with
an \emph{interpretation}: the canonical ring homomorphism $\ZZ \to \FF$
turns atomic formulas into polynomial equations over~$\FF$, logical
connectives are interpreted as usual, and quantified variables range in~$\FF$.
We write $\varphi(x_1, \dots, x_n)$ to indicate the free variables of~$\varphi$.
For any choice of $a_1, \dots, a_n \in \FF$ the substitution $\varphi(a_1,
\dots, a_n)$ produces a formula without free variables, called a
\emph{sentence}, which is either true or false in~$\FF$. Putting the
emphasis on the field as a variable, we may also say that $\FF$ satisfies
or violates the sentence and write $\FF \models \theta$ in case $\FF$
satisfies the sentence~$\theta$.
A set $X \subseteq \FF^n$ is \emph{definable} if there exists a formula
$\varphi(x_1, \dots, x_n; y_1, \dots, y_m)$ and a vector $b \in \FF^m$ of
parameters such that $X = \varphi(\FF^n; b) \defas \Set{ a \in \FF^n :
\FF \models \varphi(a,b) }$. The~coordinate projection $\pi_I(\varphi(\FF^n;b))$,
for $I \subseteq N$, of a definable set is defined by prefixing $\varphi$
with a quantifier block $\exists x_j$ for each~$j \not\in I$ and using
the same parameter~vector.

The parameters $b$ provide field-specific coefficients for use in
the polynomials defining~$X$. To~emphasize that all entries of $b$ belong
to a subring $B \subseteq \FF$, we say that $X$ is \emph{$B$-definable}.
For~example, a variety defined over $\FF$ has a unique decomposition into
absolutely irreducible components; these components may not be definable
over $\FF$ but over a finite algebraic extension.
If~$X$~is $B$-definable, then it makes sense to interpret the formula
defining $X$ in any field containing~$B$. This provides a sequence of
sets over fields of growing size, as in the Kaced--Romashchenko example
in \Cref{sec:KR} where the growth of the field leads to a violation of
the targeted inequality. Note that a $\ZZ$-definable set can be defined
entirely without parameters and hence interpreted over every~field.

\begin{convention*}
We prefer to employ a more geometric language from here on, whenever this
is beneficial. The letter $X$ denoting a definable set replaces the mention
of the defining formula $\varphi$ and the parameter vector~$b$.
The notation $X(\FF)$ for the $\FF$-rational points is extended from
varieties to definable sets and abbreviates~$\varphi(\FF^n; b)$.
\end{convention*}

\begin{definition}
Let $X$ be an $\FF$-definable set in $n$ free variables. To any finite
extension $\GG/\FF$ there is an associated \emph{coordinate random vector}
$\xi(\GG) = \xi(X; \GG)$. This random vector is $\GG^n$-valued and uniformly
supported on $\Rats_{\GG} X$, i.e., for any $a \in \GG^n$:
\[
  \Pr[\xi(\GG) = a] \defas \begin{cases}
    \sfrac1{\card{\Rats_{\GG} X}}, & \text{if $a \in \Rats_{\GG} X$}, \\
    0, & \text{otherwise}.
  \end{cases}
\]
\end{definition}

For any $I \subseteq N$ and $a_I \in \GG^I$ the marginal probability is
$\Pr[\xi_I(\GG) = a_I] = \frac{\card{\Rats_{\GG} X \cap
\pi_I^{-1}(a_I)}}{\card{\Rats_{\GG} X}}$, and thus the estimation
of the marginal entropy hinges on estimating the sizes of $\Rats_{\GG} X$
and of the fibers over its coordinate projection~$\pi_I(\Rats_{\GG} X)$.

\subsection{Measure and decomposition}

It is no accident that the asymptotic sizes of the definable sets
of cubics studied in \Cref{ex:Cubic} are of the form $\mu q^d$ for
$d \in \NN$ and $\mu \in \QQ$. The model-theoretic point of view lends
itself well to describing such uniformity properties of $\Rats_{\FF} X$
as the field~$\FF$ varies. This line of research traces back to works
of Ax~\cite{AxElementary} and his student Kiefe~\cite{KiefePoincare};
for a survey of results until the mid-1990s, see \cite{ChatzidakisSurvey}.
Kiefe proved that the logarithmic derivative of the Hasse--Weil zeta
function of any $\FF$-definable set is rational and hence the exact
point counts of $\Rats_{\GG} X$ for extensions $\GG/\FF$ are available
through a single finite object. But again, this object is very hard to
compute and the information it provides is more detailed than necessary.
Chatzidakis, van den Dries and Macintyre \cite{Definable}, refining
Kiefe's approach, proved the following uniformity result which generalizes
the Lang--Weil bound: %

\begin{theorem}[{\cite{Definable}}] \label{thm:Measure}
Consider a formula $\varphi(x_1, \dots, x_n, y_1, \dots, y_m)$.
There exist finitely many formulas $\psi_k(y_1, \dots, y_m)$, indexed by
$k \in K$, with accompanying $\mu_k \in \QQ$ and $d_k \in \NN$ such that
for every sufficiently large finite field $\FF_q$ and every $b \in \FF_q^m$:
\begin{enumerate}[noitemsep, parsep=0.3em]
\item There exists a unique $k \in K$ such that $\FF_q \models \psi_k(b)$.
\item $\FF_q \models \psi_k(b)$ if and only if $\card{\varphi(\FF_q^n; b)} =
\mu_k q^{d_k} + \Oh(\mu_k q^{d_k-\sfrac12})$.
\end{enumerate}
\end{theorem}

As a corollary to this result, the cardinality of a definable set
$X(\FF_q)$ can be estimated by an expression of the form~$\mu q^d$
up to a precision of $\Oh(q^{d-\sfrac12})$. If $\mu \neq 0$ and provided
that $q$ is large, these two numbers are necessarily unique among the
finitely many choices given by \Cref{thm:Measure}.
As~discussed in \cite[Section~4]{Definable}, the exponent~$d$ is the
dimension of the Zariski closure~of~$X(\FF_q)$. The rational coefficient~$\mu$
measures the size of~$X(\FF_q)$ relative to $\FF_q^d$. It can be less than one,
as seen in \Cref{ex:Cubic}, or greater if $X(\FF_q)$ has multiple
irreducible components of maximal dimension.
Both of these numbers depend on the field over which $X$ is interpreted.

\begin{definition}
Let $X$ be an $\FF$-definable set and $\GG/\FF$ sufficiently large,
so that \Cref{thm:Measure} guarantees the existence and uniqueness of
$\mu$ and $d$ with $\card{\Rats_{\GG} X} = \mu \card{\GG}^d +
\Oh(\mu \card{\GG}^{d-\sfrac12})$. Let~$\mu_{\GG}(X) \defas \mu$ be
the \emph{measure} and $\dim_{\GG}(X) \defas d$ the \emph{dimension}
of $X$ over~$\GG$.
\end{definition}

\begin{example}[Square root of $-1$] \label{ex:Sqrt-1}
Consider the set $X$ defined by $\exists x\colon x^2 + y^2 = 0$ and observe
$X(\FF_q) = \FF_q$ if $-1$ is a square in $\FF_q$ and $X(\FF_q) = \Set{0}$
otherwise. The existence of a square root of $-1$ in finite fields is easily
characterized in terms of their size $q = p^e$, thus:
\[
  \card{X(\FF_q)} = \begin{cases}
    q, & \text{$p \equiv 1 \pmod 4$}, \\
    q, & \text{$p \equiv 3 \pmod 4$ and $e$ even}, \\
    1, & \text{$p \equiv 3 \pmod 4$ and $e$ odd}, \\
  \end{cases}
\]
which is non-trivially periodic in $p$ (when $e$ is odd) and in~$e$
(when $p \equiv 3 \pmod 4$).
\end{example}

\Cref{thm:Measure} furnishes another essential tool: the formulas $\psi_k$.
Let $X$ be defined by a formula $\varphi(x_1, \dots, x_n; y_1, \dots, y_m)$
with parameter vector~$b \in \FF^m$. For any $I \subseteq N$ and
$J = N \setminus I$ divide the variables into $x_I$, $x_J$ and $y$ and
apply \Cref{thm:Measure} to eliminate $(x_I, y)$. The resulting formulas
$\psi_k(x_I; y)$ with parameters $y=b$ define sets $Y_k(\GG) \subseteq \GG^I$
which, for large enough field extensions $\GG/\FF$, partition~$\GG^I$ such
that if $a_I \in \Rats_{\GG}{Y_k}$, then
\[
  \card{\Rats_{\GG} X \cap \pi_I^{-1}(a_I)}
  = \card{\varphi(a_I, \GG^J; b)}
  = \mu_k \card{\GG}^{d_k} + \Oh(\mu_k \card{\GG}^{d_k-\sfrac12}).
\]
\Cref{thm:Measure} therefore implies that the fibers of a projection of a
definable set have only finitely many possible sizes. The $Y_k(\GG)$
partition~$\GG^I$ according to these fiber sizes, simultaneously for
all sufficiently large~$\GG$. We record these properties in the definition
of a \emph{fiber decomposition}. A~cartoon of a fiber decomposition is
shown in \Cref{fig:Decomp}.

\begin{definition}
Let $X$ be an $\FF$-definable set. A \emph{fiber decomposition} with respect
to~$\pi_I(X)$ is a finite family of $\FF$-definable sets $Y_k$, called
\emph{cells}, together with non-negative $\mu_k \in \QQ$ and $d_k \in \NN$,
for $k \in K$, such that for all sufficiently large~$\GG/\FF$:
\begin{enumerate}[noitemsep, parsep=0.3em]
\item $\GG^I = \bigsqcup_{k \in K} \Rats_{\GG}{Y_k}$, and
\item $\card{\Rats_{\GG} X \cap \pi_I^{-1}(a_I)} = \mu_k \card{\GG}^{d_k}
  + \Oh(\mu_k \card{\GG}^{d_k-\sfrac12})$ for each $a_I \in \Rats_{\GG}{Y_k}$.
\end{enumerate}
\end{definition}

\begin{figure}
\begin{center}
\begin{tikzpicture}[cap=round, scale=0.9]
\draw[line width=.8mm, draw=blue, fill=blue!20!white]
  (0,0) -- (0,2) -- (1,2) -- (1,1) -- (2.5,1) -- (2.5,1.5) -- (4,1.5) --
  (4,2) -- (2.5,2) -- (2.5,2.5) -- (4,2.5) -- (4,4) -- (4,0) -- (0,0);
\draw[line width=.8mm, draw=blue, fill=blue!20!white]
  (5,0) -- (5,2) -- (7,2) -- (7,1) -- (9.5,1) -- (9.5,0) -- (8.5,0) --
  (8.5,1) -- (7,1) -- (7,0) -- (5,0);

\draw[line width=.3mm, dotted] (0,-1) -- (0,5);
\draw[line width=.3mm, dotted] (1,-1) -- (1,5);
\draw[line width=.3mm, dotted] (2.5,-1) -- (2.5,5);
\draw[line width=.3mm, dotted] (4,-1) -- (4,5);
\draw[line width=.3mm, dotted] (5,-1) -- (5,5);
\draw[line width=.3mm, dotted] (7,-1) -- (7,5);
\draw[line width=.3mm, dotted] (8.5,-1) -- (8.5,5);
\draw[line width=.3mm, dotted] (9.5,-1) -- (9.5,5);

\draw[line width=.4mm, cap=round, ->] (-1,-1) -- (11,-1);
\draw[line width=1mm, draw=WildStrawberry] (0,-1) -- (1,-1);
\draw[line width=1mm, draw=Salmon] (1,-1) -- (2.5,-1);
\draw[line width=1mm, draw=WildStrawberry] (2.5,-1) -- (4,-1);
\draw[line width=1mm, fill, draw=RoyalBlue] (4,-1) circle (.6mm);
\draw[line width=1mm, draw=WildStrawberry] (5,-1) -- (7,-1);
\draw[line width=1mm, draw=Purple] (7,-1) -- (8.5,-1);
\draw[line width=1mm, draw=Salmon] (8.5,-1) -- (9.5,-1);

\draw[line width=.3mm, dotted] (-1,0) -- (11,0);
\draw[line width=.3mm, dotted] (-1,1) -- (11,1);
\draw[line width=.3mm, dotted] (-1,1.5) -- (11,1.5);
\draw[line width=.3mm, dotted] (-1,2) -- (11,2);
\draw[line width=.3mm, dotted] (-1,2.5) -- (11,2.5);
\draw[line width=.3mm, dotted] (-1,4) -- (11,4);

\draw[line width=.4mm, cap=round, ->] (-1,-1) -- (-1,5);
\draw[line width=1mm, draw=RawSienna] (-1,0) -- (-1,1);
\draw[line width=1mm, draw=Periwinkle] (-1,1) -- (-1,1.5);
\draw[line width=1mm, draw=Magenta] (-1,1.5) -- (-1,2);
\draw[line width=1mm, draw=Fuchsia] (-1,2) -- (-1,2.5);
\draw[line width=1mm, draw=RedOrange] (-1,2.5) -- (-1,4);
\draw[line width=1mm, fill, draw=Lavender] (-1,1) circle (.6mm);
\draw[line width=1mm, fill, draw=SeaGreen] (-1,2) circle (.6mm);
\end{tikzpicture}
\end{center}
\caption{A blue set $X$ and fiber decompositions with respect to its
projections onto both coordinate axes. The~axis-parallel dotted lines
delineate pieces of the projections over which the fibers have constant
size.}
\label{fig:Decomp}
\end{figure}
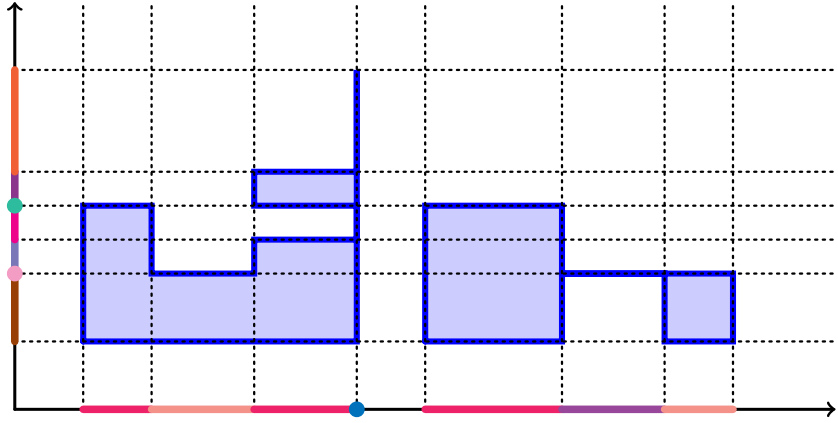

\begin{remark}
If $(Y_k : k \in K)$ is a fiber decomposition with respect to $\pi_I(X)$,
then the inverse images $X_k = X \cap \pi_I^{-1}(Y_k)$, called \emph{blocks},
are definable as well and decompose~$X$.
\end{remark}

\pagebreak %

\begin{remark}
The concept of fiber decomposition resembles the \emph{cylindrical algebraic
decomposition} used in computational real algebraic geometry \cite{AlgoReal}.
The analogy is stronger when we consider fiber decompositions $Y_k^{\ell}$
for a sequence of successive projections $\pi_\ell = \pi_{\set{1, \dots, n-\ell}}$
eliminating a single coordinate in each step. However, instead of insisting
that each projection step have definable ``cylindrical'' fibers, we here want
cells over which the fibers have an approximately~constant~size.
\end{remark}

Friedman, Haran and Jarden \cite{EffectiveCounting} gave an algebraic proof
of \Cref{thm:Measure} using their theory of Galois stratification (see
also~\cite{FieldArithmetic4ed}). Their work yields an effective procedure
for computing a fiber decomposition as well as the dimension and measure
of any definable set over a fixed field. As a corollary to their method
they observed the following crucial fact.

\begin{theorem}[{\cite{EffectiveCounting}}] \label{thm:GaloisStrat}
Let $X$ be an $\FF$-definable in $n$ variables. There is an algorithm which
computes for any $I \subseteq N$ a fiber decomposition of~$\pi_I(X)$.
Moreover, one can compute a bound $m \in \NN$, numbers $d_k \in \NN$ and
non-negative $\mu_k \in \QQ$ such that for every finite extension $\GG/\FF$:
\[
  \card{\Rats_{\GG} X} = \mu_k \card{\GG}^{d_k} + \Oh(\mu_k \card{\GG}^{d_k-\sfrac12}),
\]
where $k \equiv [\GG : \FF] \pmod m$.
\end{theorem}

This confirms that the periodicity of cardinalities observed in \Cref{ex:Sqrt-1}
was no accident. Using the data computed by \Cref{thm:GaloisStrat} one
can find an estimate for the cardinality of $\Rats_{\GG}{X}$ by reducing
$[\GG : \FF] \bmod m$ without having to do computations in the (large)
field~$\GG$. It also shows that every pair $(d_k, \mu_k)$ which appears
for some extension of $\FF$ appears for arbitrarily large~fields.

\begin{remark}
The error term in the cardinality estimates in \Cref{thm:Measure} and
\Cref{thm:GaloisStrat} can be bounded absolutely in terms of $\card{\GG}$
and computable constants; cf.~\cite[Theorem~6.4]{EffectiveCounting}.
\end{remark}

The results of Friedman, Haran and Jarden are based on a procedure called
\emph{Galois stratification} which is described in detail in the monograph
\cite{FieldArithmetic4ed}. We are not aware of any implementation of this
procedure and developing one appears to be a major undertaking. %
Nevertheless we can demonstrate its core ideas by repeating \Cref{ex:Cubic}
using facts from Galois theory \cite{Morandi}.

\begin{example}
Fix a finite field $\FF$ of large characteristic and consider the variety~$V$
defined by $f = x^3 + ax^2 + bx + c \in \FF[a,b,c,x]$. The pieces of a fiber
decomposition eliminating~$x$ from~$V$ are given by stratifying the triples
$(a,b,c)$ according to the number of rational roots of the specialization
$f(a,b,c) \in \FF[x]$. The discriminant of $f$ with respect to $x$ is a
non-trivial polynomial in $a, b, c$. Hence, the locus where $f(a,b,c)$ is
inseparable is of lower dimension and may be ignored.

Let $\Omega$ be the splitting field of $f$ over $\FF(a,b,c)$; its Galois group
is the symmetric group~$S_3$. %
The specialization $f(a,b,c)$ --- assumed to be separable --- also defines
a Galois extension over~$\FF$ with cyclic Galois group~$G(a,b,c)$. The number
of rational roots of $f(a,b,c)$ over~$\FF$ is determined by the splitting type
of $f(a,b,c)$~in~$\FF[x]$. This in turn corresponds to the conjugacy class of
$G(a,b,c)$~in~$S_3$.
By an analogue of the Chebotarev density theorem \cite[Theorem~4.4]{EffectiveCounting},
the density of the triples with a given conjugacy class $\CC C$ (whose
splitting type appears full-dimensionally) can be computed simply~via
\[
  \frac{\card{\CC C}}{[\Omega : \FF(a,b,c)]} = \frac{\card{\CC C}}{6},
\]
where we have used that $\FF$ is algebraically closed in~$\Omega$.
This yields at once the following table:
\begin{center}
\setlength{\tabcolsep}{1.5em}
\renewcommand{\arraystretch}{1.3}
\begin{tabular}{c||c|c|c|c|c}
Splitting type  & $[1,1,1]$  & $[1,2]$    & $[3]$       & $[1, 1^2]$ & $[1^3]$  \\ \hline
Conjugacy class & $\id$      & $(1 \; 2)$ & $(1\;2\;3)$ & $\blank$   & $\blank$ \\ \hline
Density         & $\sfrac16$ & $\sfrac36$ & $\sfrac26$  & $0$        & $0$      \\ \hline
Fiber size      & $3$        & $1$        & $0$         & $2$        & $1$
\end{tabular}
\end{center}
and confirms the calculations in \Cref{ex:Cubic}.
\end{example}

\subsection{Entropy profiles and algebraic independence}

The results on dimension and measure enable us to compute the entropy
profiles of the coordinate random vector of any definable set. %

\begin{theorem} \label{thm:Unif}
Let $X$ be an $\FF$-definable set in $n$ free variables and $\xi$ the
corresponding coordinate random vector. For $I \subseteq N$ let
$(Y_k : k \in K)$, be a fiber decomposition with repsect to $\pi_I(X)$
and set $X_k = X \cap \pi_I^{-1}(Y_k)$. For large enough $\GG/\FF$,
the entropy profile satisfies
\[
  h_{\xi(\GG)}(I) = \sum_{\dim_{\GG}(X_k) = \dim_{\GG}(X)}
    \frac{\mu_{\GG}(X_k)}{\mu_{\GG}(X)} \log\left(\frac{\mu_{\GG}(X) \mu_{\GG}(Y_k)}{\mu_{\GG}(X_k)} \card{\GG}^{\dim_{\GG}(Y_k)}\right) +
      \Oh\left(\frac{\log{\card{\GG}}}{\sqrt{\card{\GG}}}\right).
\]
The leading term which does not vanish asymptotically can be effectively
computed from a defining formula for~$X$ and is periodic in the extension
degree~$[\GG:\FF]$.
\end{theorem}

\begin{proof}
The probability of an event $a_I \in \GG^I$ is determined by the relative
size of its fiber in $\Rats_{\GG} X$:
\[
  \Pr[\xi_I(\GG) = a_I] = \frac{\card{\Rats_{\GG} X \cap \pi_I^{-1}(a_I)}}{\card{\Rats_{\GG} X}}.
\]
Whenever $a_I \in \Rats_{\GG}{Y_k}$, we have $\card{\Rats_{\GG} X \cap
\pi_I^{-1}(a_I)} = \frac{\card{\Rats_{\GG}{X_k}}}{\card{\Rats_{\GG}{Y_k}}}$.
\Cref{thm:Measure} guarantees that every event $a_I$ belongs to exactly one
of the pieces $\Rats_{\GG}{Y_k}$ in the fiber decomposition.
For brevity, we adopt the following notation:
\begin{center}
\setlength{\tabcolsep}{1.5em}
\renewcommand{\arraystretch}{1.3}
\begin{tabular}{c||c|c|c}
                     & $X$   & $X_k$   & $Y_k$ \\ \hline
$\mu_{\GG}(\blank)$  & $\mu$ & $\mu_k$ & $\nu_k$ \\ \hline
$\dim_{\GG}(\blank)$ & $d$   & $d_k$   & $e_k$
\end{tabular}
\end{center}
This suffices to compute the marginal entropy asymptotically:
\begin{align*}
  \H{\xi_I(\GG)}
    &= \sum_{k \in K} \sum_{a_I \in \Rats_{\GG}{Y_k}} \frac{\card{\Rats_{\GG} X \cap \pi_I^{-1}(a_I)}}{\card{\Rats_{\GG} X}}
       \log\left(\frac{\card{\Rats_{\GG} X}}{\card{\Rats_{\GG} X \cap \pi_I^{-1}(a_I)}}\right) \\
    &= \sum_{k \in K} \frac{\card{\Rats_{\GG}{X_k}}}{\card{\Rats_{\GG} X}}
       \log\left(\frac{\card{\Rats_{\GG} X} \card{\Rats_{\GG}{Y_k}}}{\card{\Rats_{\GG}{X_k}}}\right) \\
    &= \sum_{k \in K} \frac{\mu_k}{\mu q^{d - d_k}}
       \log\left(\frac{\mu \nu_k}{\mu_k} \card{\GG}^{d+e_k-d_k}\right) +
       \Oh\left(\frac{\log{\card{\GG}}}{\sqrt{\card{\GG}}}\right).
\end{align*}
Elementary properties of the point count $\card{\blank}$ imply that $d_k \le d$
for every~$k$ and equality is achieved at least once. The only summands which
matter asymptotically are those in which~$d_k = d$:
\begin{equation}
  \label{eq:EntX}
  \H{\xi_I(\GG)} = \sum_{d_k = d} \frac{\mu_k}{\mu}
    \log\left(\frac{\mu \nu_k}{\mu_k} \card{\GG}^{e_k}\right) +
    \Oh\left(\frac{\log{\card{\GG}}}{\sqrt{\card{\GG}}}\right),
\end{equation}
which proves the formula. The computability and periodicity follow at once
from \Cref{thm:GaloisStrat}
\end{proof}

\begin{remark}
\Cref{thm:Unif} shows that adding or removing lower-dimensional sets
to or from $X$ does not influence the leading terms in the entropy
profile. Moreover, it is clear that $\sum_{d_k=d} \mu_k = \mu$ and
hence \eqref{eq:EntX} is a convex combination of the logarithmic~terms.
\end{remark}

\begin{remark}
Each entry in $h_{\xi(\GG)}$ is periodic in the extension degree~$[\GG:\FF]$.
Knowing all period lengths and the values taken by the components, it is an
easy exercise in modular arithmetic to determine the distinguished entropy
profiles which appear for arbitrarily large finite~extensions.
\end{remark}

\begin{corollary} \label{cor:Variety}
Let $\Rats_{\FF} X$ be an $\FF$-definable irreducible variety given by a
prime ideal $\SR I \subseteq \FF[x_N]$. Then there exists a tower of finite
fields $\FF = \GG_0 \subseteq \GG_1 \subseteq \dots$ with
\[
  \lim_{n \to \infty} \frac1{\log{\card{\GG_n}}} h_{\xi(\GG_n)}(I)
    = \dim \pi_I(\Rats_{\ol{\FF}} X), \;\; \text{for every $I \subseteq N$}.
\]
\end{corollary}

\begin{proof}
Consider the irreducible components of~$\Rats_{\ol \FF} X$. They are defined
by finitely many equations with coefficients in~$\ol{\FF}$ and hence they are
definable in a finite extension~$\GG_1/\FF$. The irreducible components of
every projection $\pi_I(\Rats_{\ol \FF} X)$ are then definable over~$\GG_1$
as well. Take $\GG_n = \GG_1^n$.
Since at least one of the absolutely irreducible components achieves the
dimension~$\dim \pi_I(\Rats_{\ol \FF} X)$, the Lang--Weil bound implies that
$\dim_{\GG_n} \pi_I(X) = \dim \pi_I(\Rats_{\ol \FF} X)$ for all $I \subseteq N$
and $n \ge 1$. Therefore, the dimensions computed by \Cref{thm:Measure}
coincide with the geometric dimensions.

Set $\GG = \GG_n$ for $n$ large enough, fix a projection $\pi_I$ and adopt
the notation of the proof of \Cref{thm:Unif} where $(Y_k : k \in K)$ is a
fiber decomposition and $(X_k : k \in K)$ the inverse images in~$X$.
The~rescaled entropy is
\begin{equation}
  \label{eq:Hq}
  \frac1{\log{\card{\GG}}} \H{\xi_I(\GG)} = \sum_{d_k = d} \frac{\mu_k}{\mu}
  e_k + \Oh\left(\frac1{\log{\card \GG}}\right),
\end{equation}
where $\sum_{d_k = d} \mu_k = \mu$. Hence, the goal is to show that
$e_k = \dim \pi_I(\Rats_{\ol\FF} X)$ whenever $d_k = d$.

Using induction, it suffices to treat the case when a single coordinate
$x_j$ is projected away, i.e., $I = N \setminus \Set{j}$. Pick a full-dimensional
block~$X_k$. By the properties of a fiber decomposition, every fiber
$\Rats_{\GG} X \cap \pi_I^{-1}(a_I)$ for $a_I \in Y_k = \pi_I(X_k)$ has
the same dimension; subtracting the fiber dimension from $d_k$ yields~$e_k$.
There are two cases to consider:
\begin{paraenum}[label=(\alph*)]
\item \label{cor:Variety:Free}
If $x_j$ is algebraically independent of the remaining variables $x_I$
over~$\FF$, then this remains true over the algebraic extension $\GG/\FF$.
Since $X$ imposes no relation on $x_j$ in terms of $x_I$, there are
$\card{\GG}$ preimages for any $a_I \in \Rats_{\GG}{Y_k}$.

\item \label{cor:Variety:Finite}
Otherwise $x_j$ is algebraically dependent on $x_I$ over~$\FF$.
Let $I' \subseteq I$ be inclusion-minimal with the property that $x_{I'}$ is
algebraically independent and $x_j$ is algebraically dependent on $x_{I'}$,
i.e., $I' \cup \Set{j}$ is a circuit of the algebraic matroid of~$\Rats_{\FF} X$.
Since $\SR I$ is prime, the elimination ideal $\SR I \cap \FF[x_j, x_{I'}]$
is prime and principal, hence generated by an irreducible polynomial~$f$.
Write $f = \sum_{i=0}^r f_i x_j^i$ with $f_i \in \FF[x_{I'}]$. For each
$a_I \in \Rats_{\GG}{Y_k}$ there are at most as many values~$a_j$ for the
preimage as there are solutions to $f(a_j, a_{I'}) = 0$.
We claim that there exists $a_I \in \Rats_{\GG}{Y_k}$ such that
$f(x_j, a_{I'}) \not= 0 \in \GG[x_j]$. It then follows that $f(a_j, a_{I'}) = 0$
has at most $r$~solutions and the special fiber above $a_I$ is zero-dimensional.
Since $Y_k$ is a cell in a fiber decomposition, all fibers above it are
zero-dimensional.
To prove the claim, first note that $f_0 \not= 0$ since otherwise $f$ would
be divisible by~$x_j$. By the independence of $x_{I'}$ there must be another
non-zero coefficient~$f_\ell$ so that the variable $x_j$ appears in~$f$.
We may pick~$f_\ell$ coprime to $f_0$; for if that were impossible, then~$f$
would be reducible. The locus $Z$ of all $a_I \in \pi_I(\Rats_{\ol\FF} X)$
such that $f(x_j, a_{I'}) = 0$ has codimension at least two since $x_{I'}$ are
algebraically independent and must satisfy the two coprime conditions $f_0 =
f_\ell = 0$. But then the preimage of~$Z$ can have dimension at most $d-2+1 =
d-1$ and any cell $\Rats_{\GG}{Y_k}$ contained in $\Rats_{\GG} Z$ does not
satisfy~$d_k = d$.
\end{paraenum}

Inductively, this shows that for each block $X_k$ with $d_k = d$, the
difference $d - e_k$ is precisely the number of times case~\ref{cor:Variety:Free}
applies during the projections, which is the transcendence degree of
$\FF(x_N)/\FF(x_I)$. This yields $e_k = \dim \pi_I(\Rats_{\ol\FF} X)$
as required.
\end{proof}

\begin{example} \label{ex:Reducible}
If the variety is reducible and the components are arranged in a special
position with respect to the coordinate hyperplanes, the conclusion of
\Cref{cor:Variety} does not hold. Consider the variety $V$ defined
by~$xy = 0$ over a finite field~$\FF$. Its projection to the $y$-axis
has a fiber decomposition consisting of two parts: $Y_1 = \Set{0}$ and
$Y_2 = \FF \setminus \Set{0}$ with the following~data:
\begin{center}
\setlength{\tabcolsep}{1.5em}
\renewcommand{\arraystretch}{1.2}
\begin{tabular}{c||c|c|c|c|c}
                     & $V$ & $Y_1$ & $Y_2$ & $V \cap \pi_y^{-1}(Y_1)$ & $V \cap \pi_y^{-1}(Y_2)$ \\ \hline
$\mu_{\FF}(\blank)$  & $2$ & $1$   & $1$   & $1$   & $1$   \\ \hline
$\dim_{\FF}(\blank)$ & $1$ & $0$   & $1$   & $1$   & $1$
\end{tabular}
\end{center}
Hence by \Cref{thm:Unif} we have $\frac1{\log{\card{\FF}}} h_{\xi(\FF)}(y) \to \frac{1}{2}$
which does not match $\dim \pi_y(V) = 1$.
\end{example}

Notice that in this result algebraic independence of coordinate functions in
the limit is explained through diminishing stochastic dependence among the
coordinate random variables relative to the growing field size. Hence, at
least on a suitable sequence of fields, $h_\xi$ is a refinement of the
algebraic matroid of~$X$. This connection was already mentioned in \Cref{sec:KR}.
It can be used to derive the following theorem, originally due to Matúš.

\begin{corollary}[{\cite{MatusAlgebraic}}]
Algebraic matroids are almost-entropic.
\end{corollary}

\begin{proof}
Let $\SI M = (N, r)$ be an algebraic matroid over a field~$\FF$. First
suppose that $\FF$ has characteristic zero. Then~$\SI M$ is also linear
over a field of characteristic zero (using \emph{derivations} as in
\cite[Section~6]{Ingleton}). %
The condition that a fixed matroid is linear over a field can be expressed
as a sentence in the first-order language of rings.
The~model-theoretic Lefschetz principle \cite[Corollary~2.2.10]{Marker}
thus implies that $\SI M$ is linear over all but finitely many~$\ol{\FF_p}$.
Since a linear representation is a matrix with finitely many entries,
all algebraic over $\FF_p$, the matroid is linear over a finite field.
In turn, this gives an algebraic representation over a finite field
and it suffices to handle the case in which $\FF$ has characteristic $p > 0$
in the remainder~of~the~proof.

Now suppose that $\SI M$ is represented by a collection $(x_i : i \in N)$
of elements in a field extension~$\FF'/\FF$ where $\FF$ has characteristic
$p > 0$. By results of Lindström and Piff (see~\cite{LindstromFt}) we may
suppose that $\FF = \FF_p$.  The ring homomorphism $\FF_p[t_1, \dots, t_n]
\to \FF'$ given by $t_i \mapsto x_i$ has a prime kernel $\SR I$ which is
generated by finitely many polynomials with coefficients in~$\FF_p$.
They define an irreducible variety~$\Rats_{\FF_p} X$. The algebraic matroid
keeps track of the Krull dimensions of the elimination ideals, i.e.,
$r(I) = \dim \FF_p[t_I]/(\FF_p[t_I] \cap \SR I) = \dim \pi_I(\Rats_{\ol{\FF_p}} X)$.
The claim follows from \Cref{cor:Variety} by noticing that $r = \lim_{n\to\infty}
\frac1{\log{\card{\GG_n}}} h_{\xi(\GG_n)}$ is almost-entropic.
\end{proof}

\begin{remark}
Our proof fills a small hole in Matúš's, where he asserts incorrectly that
if a coordinate $x_j$ is algebraically dependent on the remaining coordinates
on $\Rats_{\ol\FF}{X}$, then the fibers under the projection eliminating~$x_j$
are always zero-dimensional. This problem is addressed in case~\ref{cor:Variety:Finite}
in the proof of \Cref{cor:Variety} by showing that one-dimensional fibers
occur only on a subset of the projection of codimension at least two.
Hence, this phenomenon concerns a lower-dimensional part of~$\Rats_{\ol\FF}{X}$
and does not interfere with the asymptotic calculations.
\end{remark}

\begin{remark}
It would be interesting to relate the normalized entropy profile of a
variety as given by \eqref{eq:Hq} to those of its irreducible components;
in particular whether the former is always in the convex hull of the
latter. This is true for simple examples like $V(xy)$, $V((x-y)(x+y))$
and $V((y-x)(y-x^2))$. Note that the fiber decomposition does not
distribute over~components.
\end{remark}

\subsection{Linear congruences and monomial maps}
\label{sec:Linear}

In general, entropy profiles of definable sets can be computed using
Galois stratification. However, we know of no implementation of this
algorithm. This section briefly considers the special case of toric
varieties, and the associated linear algebra, for which the required
algorithms are widely implemented.

Fix a matrix $A \in \ZZ^{n \times d}$. For each choice of a residue ring
$R = \ZZ/m$, we may consider the random vector $\xi_A(R)$ on $R^n$ supported
uniformly on the $R$-submodule~$\im(\ol A)$. This is the submodule generated
by the columns of the matrix~$\ol A$ arising from reduction of $A$~modulo~$m$.
The~goal is to compute the entropy profile of $\xi_A(R)$ as a function of
$A$~and~$m$. For~each $I \subseteq N$, denote by $A_I$ the submatrix
consisting of the rows indexed by~$I$. Then we have $\im(\ol A_I) =
\pi_I(\im \ol A)$. Since the projection is an $R$-module homomorphism, all
its fibers have the same size $\card{\ker \pi_I}$ which shows that $\xi_A(R)$
is quasi-uniform.

The key to computing $\im(\ol A)$ is the \emph{Smith normal form}.
Recall (e.g., from \cite{StanleySNF}) that every $A \in \ZZ^{n \times d}$
has a Smith normal form $S \in \ZZ^{n \times d}$ whose diagonal entries
$s_1, \dots, s_k$, $k = \min\{n, d\}$, satisfy $s_i \mid s_{i+1}$, and which
is zero outside of its diagonal. Moreover, $S = TA\,U$ for invertible matrices
$T \in \ZZ^n$ and $U \in \ZZ^d$. Reducing the matrix equation $S = TA\,U$
modulo~$m$ gives a Smith normal form of $\ol A$ over~$R$.
The transformations $\ol{T}$ and $\ol{U}$ remain bijective over~$R$, so
$\card{\ker \ol A} = \card{\ker \ol S}$. But $\ol S$ is diagonal, so its
kernel is isomorphic to the direct sum of $R^{d-k}$ and the kernels of
the scalar multiplication maps in its rows. Multiplication by
$s \in \ZZ$ annihilates precisely $\gcd(m,s)$ elements of $\ZZ/m$ (where
$\gcd(m,0) = m$). Then the isomorphism theorem computes the image size as
\begin{equation}
  \label{eq:SubmodSize}
  \card{\im \ol A} = \prod_{i=1}^k \frac{m}{\gcd(m, s_i)}.
\end{equation}
The same works for all submatrices~$A_I$, completing the description of
the entropy profile depending on $m$ and the Smith normal forms of all
submatrices~$A_I$,~$I \subseteq N$.

\begin{example}
This procedure is easy to implement in \Macaulay2:
\begin{minted}{macaulay2}
entropy = (A,m) -> (
  S := (smithNormalForm A)#0;
  r := #entries A; c := #entries transpose A;
  diags := select(apply(0 .. min(r, c)-1, i -> S_(i,i)), s -> s != 0);
  product apply(toList diags, s -> m/gcd(m,s))
);

entropyProfile = (A,m) ->
  apply(subsets(toList(0 .. (#entries A)-1)), I ->
    I => if #I == 0 then 1 else entropy(A^I, m));

A = random(ZZ^4, ZZ^5);
netList entropyProfile(A, 7^3)
\end{minted}
The most common outcome is a uniform polymatroid (where $h(I) = \sum_{i\in I}
h(i)$) but other outcomes are possible. For instance the matrix on the left
yields the almost-entropic point on the~right
\begin{equation*}
  \begin{pmatrix}
    1 & 0 & 2 & 3 & 0 \\
    2 & 9 & 7 & 7 & 7 \\
    9 & 3 & 3 & 3 & 0 \\
    2 & 2 & 7 & 7 & 7
  \end{pmatrix} \implies\;
  \begin{gathered}
    h(1) = h(2) = h(3) = h(4) = 3, \\
    h(1,2) = h(1,3) = h(1,4) = h(2,3) = h(3,4) = 6, \; h(2,4) = 5, \\
    h(1,2,3) = h(1,3,4) = 9, \; h(1,2,4) = h(2,3,4) = 8, \\
    h(1,2,3,4) = 11,
  \end{gathered}
\end{equation*}
where all logarithms are in base~$7$.
\end{example}

\begin{example}
In case $m = p$ is prime, the above situation simplifies to linear algebra
over the field $R = \FF_p$. Then \eqref{eq:SubmodSize} yields $\card{\im \ol A}
= p^{\rank \ol A}$ as expected. When choosing base $p$ for the logarithm,
the entropy profile is precisely the \emph{linear matroid} associated to
the rows of~$\ol{A}$.
\end{example}

\begin{example}
The matrix $A$ determines a monomial map
\begin{equation}
  \label{eq:ToricMap}
  t = (t_1, \dots, t_d) \mapsto (t^{a_1}, \dots, t^{a_n}),
\end{equation}
where $a_1, \dots, a_n$ are the rows of~$A$ and $t^a \defas \prod_{j=1}^d t_j^{a_j}$
is the monomial-vector notation. We are interested in the uniform distribution
supported on image of the algebraic torus $(\FF_q^\times)^d$ under this map.
Recall that the multiplicative group $\FF_q^\times$ is cyclic and let $g$ be
a generator. Then~every $t_j \in \FF_q^\times$ can be written uniquely as
$t_j = g^{x_j}$ for some $x_j \in \ZZ/m$ where $m = q-1$. This~establishes
an isomorphism $\ZZ/m \to \FF_q^\times$. Observe that a monomial rewrites
to $t^a = g^{a \cdot x}$ and therefore conjugation with this isomorphism
transforms the monomial map~\eqref{eq:ToricMap} to the $\ZZ/m$-linear map
$x \mapsto \ol A x$ treated above.

Let $s$ be the least common multiple over all diagonal entries of the Smith
normal forms of $A_I$, $I \subseteq N$. By Dirichlet's theorem on primes
in arithmetic progressions, there exist infinitely many primes of the form
$p = 1 + ks$, $k \ge 1$. For $m = p-1$, every factor in \eqref{eq:SubmodSize}
simplifies to $\gcd(m, s_i) = s_i$ when $s_i \neq 0$ (when $s_i = 0$, treat
it as~$m$ instead). This gives one particular subsequence of distributions
$\xi_A(\FF_p^\times)$ whose entropy profiles are easy to estimate and they
coincide with the uniform distributions on the monomial images of the
algebraic tori~$(\FF_p^\times)^d$.
\end{example}

\begin{remark}
Fix any $R = \ZZ/m$ and let $K_I = \ker \ol{A_I} \subseteq R^d$.
These submodules have a natural semilattice structure induced from $2^N$
via $K_I = \bigcap_{i \in I} K_i$. The entropy of the marginal on~$I$ is
the logarithm of the index $[R^d : K_I]$. As such, entropy profiles defined
via linear congruences fall into the category of abelian group-representable
profiles for which Chan proved that they satisfy the Ingleton
inequality~\cite[Theorem~3.4]{ChanAspects}.
\end{remark}

\section{Information-theoretic post-processing techniques}
\label{sec:Extension}

The main result of \Cref{sec:Unif} exhibits a class of computable and
highly structured entropy profiles. By normalizing the base of the
logarithm to the growing field size, one obtains in the limit
almost-entropic points which include the class of algebraic matroids.
These limit points can be exactly represented and manipulated in a
computer algebra system.
In particular it is possible to decide whether such an almost-entropic
point violates a proposed information inequality or even to use the
sequence to prove that a valid conditional information inequality is
essentially conditional (as seen in \Cref{sec:KR}).
This section provides an overview of some information-theoretic operations
which can be carried out on these examples symbolically and yield an even
broader class of almost-entropic~points.

The almost-entropic region $\Aent_N$, for being a closed set, enjoys a
number of closure properties which $\Ent_N$ lacks (in many but not all~cases).
Results of Matúš and Csirmaz \cite{MatusTwocon,MatusCsirmaz} show that
$\Aent_N$ is closed under convolution with \emph{modular} polymatroids,
i.e., if $h \in \Aent_N$ and $m\colon 2^N \to \RR$ satisfies $\CId{I:J|K}{m} = 0$
for all disjoint $I, J, K \subseteq N$,~then
\[
  (h * m)(I) \defas \min_{J \subseteq I} \Set{ h(J) + m(I \setminus J) }
\]
defines a function $h*m\colon 2^N \to \RR$ which also belongs to~$\Aent_N$.
This entails the closedness of $\Aent_N$ under a number of matroid-theoretic
operations such as principal extensions, free expansions (where applicable)
and tightening to eliminate private information. The ramifications of this
closedness under modular convolution are certainly not yet fully~explored.

On the other hand, there is a family of lift-and-project techniques which
are called \emph{extension properties}, following \cite{ExtensionProps}.
They are theorems of the form:
\begin{equation}
  \text{For each $h \in \Aent_N \cap \CC L$ %
  there exists an $\hat h \in \Aent_M \cap \hat{\CC L}$
  (for $M \supseteq N$) with $\hat{h}|_N = h$.}
\end{equation}
In each instance of an extension property, the linear spaces $\CC L$
and $\hat{\CC L}$ are concrete and usually defined by functional dependence
or conditional independence predicates. The size of the extension
$M \setminus N$ is also concrete. Usually, by factoring, one may assume
it to be a \emph{one-point extension}.
Hence, if $h$ is almost-entropic and satisfies some linear conditions,
then it can be lifted or \emph{extended} into a polymatroid on a larger
ground set which is almost-entropic and satisfies additional linear~conditions.
Extension properties encapsulate ``coding lemmas'' in information
theory which provide an operational characterization or direct construction
of the variables by which the system is extended. The intuition behind their
constructions makes extension properties well-suited for the design of
random variables with special information-theoretic characteristics.

The most famous extension property for $\Ent_N$ is the \emph{Copy lemma};
see \cite{Selfadhe} for an abstract treatment and historical references.

\begin{named*}{Copy lemma}
Let $h \in \Ent_N$ and $L \subseteq N$. First, choose another index~set $N'$
with $\card{N'} = \card{N}$ and $N' \cap N = L$; then set $M = N \cup N'$ and
let $\tau\colon M \to M$ be a bijection that sends $\tau(N') = N$ and fixes~$L$
point-wise. There exists~$\hat h \in \Ent_M$ extending $h$~with
\begin{gather}
  \hat{h}|_{N'} = (\tau^* h)|_{N'}, \; \text{and} \\
  \label{eq:CopyCI} \CId{N:N'|L}{\hat h} = 0.
\end{gather}
\end{named*}

Equation \eqref{eq:CopyCI} states that $N$ is independent of its copy $N'$
given their common restriction~$L$.
This~theorem is used in \cite{DFZ11} to generate new information
inequalities: namely, the key feature eq.~\eqref{eq:CopyCI} of the Copy~lemma
embedding puts $\hat h$ in a special position which is \emph{more extreme}
in $\Aent_M$ than $h$ was in $\Aent_N$;
the Shannon inequalities on $\Aent_M$ interact with this equation
and imply additional inequalities for~$\hat h$, some of which put
extra constraints on the restriction~$h = \hat{h}|_N$. Since $h \in \Ent_N$
is arbitrary, these new constraints are valid information inequalities
and may~be~non-Shannon.

Extension properties can be used in the other direction as well.
Instead of having $h \in \Aent_N$ arbitrary to derive new general
constraints from the extension, one can pick a concrete almost-entropic
point and extend it to obtain another almost-entropic point which can
be projected down in many different ways, leaving ``traces'' of
almost-entropic points in the original space. One~of these new
specimen may be the sought-after counterexample.
We note that the linear constraints in the extension property may not
uniquely determine a point but a whole polyhedron of possible locations
for the extension. In case of the Copy lemma, the underlying construction
is known as the \emph{conditional product} \cite[Section~II.C]{CondIngleton}
and can be carried out on the level of probability distributions (hence
it preserves entropicness, not just almost-entropicness). It is an easy
exercise to transfer it to our setting of definable sets which gives
a precise entropy profile satisfying the constraints of the Copy
lemma. Otherwise, the extended profile is at least partially defined
and this information may be enough to deduce the existence of an
interesting~example.

The following extension plays a major role in the results of Matúš
\cite{MatusTwocon} on convolution. Kaced and Romashchenko also applied
it to the example from \Cref{sec:KR} to get an almost-entropic point
which violates a conditional Ingleton inequality which is known to be
valid for entropic points, proving that conditional information
inequalities are not ``continuous'' as a result of the porous boundary
structure of $\Ent_N$; cf.~\cite[Section~V.A]{KR}.

\begin{named*}{{Slepian--Wolf \cite[Theorem~3]{MatusTwocon}}}
For $h \in \Aent_N$ and $L \subseteq N$, let $I = N \setminus L$ and
$z \not\in N$. For~every $\alpha \ge 0$ there exists a one-point extension
$\hat h \in \Aent_{N\cup z}$ such that: %
$\CId{z|I}{\hat h} = 0$ and
$\hat{h}(K \cup z) = \min\Set{ \alpha + h(K),\, h(I \cup K) }$
for every~$K \subseteq L$.
\end{named*}

In this result, the additional component $z$ is a function of~$I$.
Interesting special cases arise when the minimum is attained twice.
In particular, for $\alpha = \CId{I|L}{h}$ we get that $\hat h(z) =
\CId{I|L}{h}$ and $\hat h(L \cup z) = h(N)$. These two relations
imply $\CId{I|L\cup z}{\hat h} = 0$ and so~$I$ is also recoverable
as a function of $z$ and $L$ together.

\pagebreak %

Finally, the Ahlswede--Körner lemma is well-known and its potential
as an extension property in the context of algebraic matroids was
also recognized in~\cite{ExtensionProps}.

\begin{named*}{{Ahlswede--Körner \cite[Lemma~2]{KacedEquivalence}}}
For $h \in \Aent_N$ and $L \subseteq N$, let $I = N \setminus L$ and
$z \not\in N$. There exists a one-point extension $\hat h \in \Aent_{N\cup z}$
such that: %
$\CId{z|L}{\hat h} = 0$ and
$\CId{K|z}{\hat h} = \CId{K|I}{h}$
for every~$K \subseteq L$.
\end{named*}

The lemma guarantees the existence of a random variable~$z$ which
is a function of~$L$ but, surprisingly, conditioning on $z$ in the
extension has the same effect on subvectors of~$L$ as conditioning
on~$I$, even though a~priori $z$ is unrelated to~$I$. In a sense,
$z$ extracts whatever information $I$ holds about~$L$.

The above-mentioned extension properties are valid for almost-entropic
points: given $h \in \Aent_N$ as input, one gets an extension $\hat h
\in \Aent_M$. It would be interesting to check which of them also hold
in the narrower class of distributions introduced in this paper.
Formally, let $\Aalg_N$ denote the smallest closed convex cone containing
the set of entropy profiles of coordinate random vectors of definable sets
over a finite field, as per \Cref{thm:Unif}; for want of a better term,
we call them \emph{almost-algebraic}. Which of the above extension
properties give $\hat h \in \Aalg_M$ provided $h \in \Aalg_N$?
Answering this question is a win-win scenario.
If an extension property holds in the algebraic setting, this corresponds
to a universal geometric construction on definable sets and provides a
combinatorial property to distinguish them from arbitrary polymatroids;
cf.~\cite{DressLovasz}. Otherwise, the application of the property yields
almost-entropic points outside the scope of the almost-algebraic region~$\Aalg_N$
and thus these points are even more interesting as examples.

All of the above post-processing techniques can and should be formulated
using the language of linear maps on polyhedra in the space~$\RR^{2^N}$.
Moving to this common ground allows an implementation to not only iterate
on a single extension property, as in \cite{DFZ11}, but to mix all of
these techniques freely.
The~polyhedral encoding also provides standard formats for (conditional)
information inequalities, their proofs, interesting examples and even
proofs of essential conditionality, so that this valuable research data
can be shared, maintained~and~reused.

\section*{Acknowledgements}

\setlength{\intextsep}{5pt}%
\setlength{\columnsep}{5pt}%
\begin{wrapfigure}{R}{0.2\linewidth}
\vspace{-.5\baselineskip}%
\centering%
\href{https://doi.org/10.3030/101110545}{%
\includegraphics[width=0.9\linewidth]{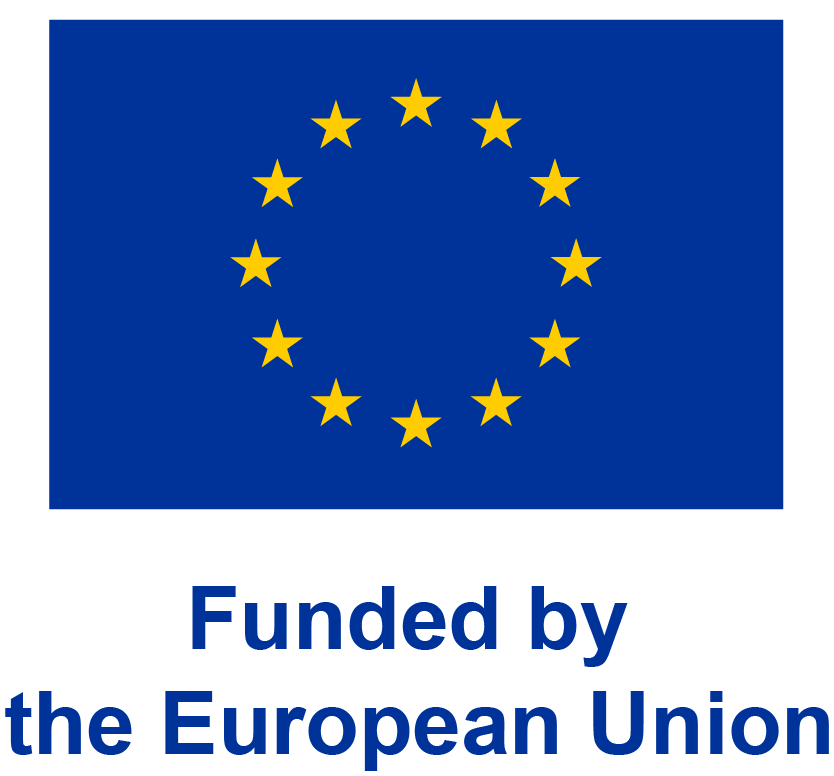}%
}
\end{wrapfigure}
I would like to thank Andrei Romashchenko for drawing my attention to
the paper of Gómez, Mejía and Montoya at the Dagstuhl seminar 22301
``Algorithmic Aspects of Information Theory''.
I am also grateful to Sachi Hashimoto and Sameera Vemulapalli for
discussions on the number-theoric aspects of this topic.
This research was funded by the European Union's Horizon 2020 research
and innovation programme under the Marie Skłodowska-Curie grant agreement
No.~101110545.

\bibliographystyle{tboege}
\bibliography{unif}

\let\etalchar\undefined
\bibliographystylesoft{tboege}
\bibliographysoft{unif}

\end{document}